\documentclass[journal]{IEEEtran}
\usepackage{graphicx,epsfig,subfigure}
\usepackage{amsmath,mathrsfs} 
\usepackage{algorithm}
\usepackage[noend]{algpseudocode}
\makeatletter
\def\BState{\State\hskip-\ALG@thistlm}
\makeatother
\usepackage{amssymb} 
\usepackage{amsthm}
\usepackage{amsfonts}
\usepackage{mathtools}
\usepackage{multirow}
\usepackage{mathtools} 
\usepackage{cite}
\usepackage{epstopdf}
\usepackage{ifthen}
\usepackage{bbm}
\usepackage{setspace}

\newcommand{\R}{ \mathbb{R}}
\newcommand{\Rm}{\mathbb{R}^M}
\newcommand{\Rd}{\mathbb{R}}

\newcommand{\Lop}{{\rm L}}

\newcommand{\bl}{{\mathcal{X}_{2}}}
\newcommand{\bltv}{{\mathcal{X}_{ 1}}}

\newcommand{\nl}{\mathcal N_{\rm L}}

\def\V#1{{\boldsymbol{#1}}}         
\def\Spc#1{{\mathcal{#1}}}  
\def\M#1{{\bf{#1}}}  
\def\Op#1{{\mathrm{#1}}}  

\newtheorem{theorem}{Theorem}
\newtheorem{proposition}[theorem]{Proposition}
\newtheorem{lemma}[theorem]{Lemma}

\newtheorem{definition}[theorem]{Definition}

\newtheorem{assumption}{Assumption}

\usepackage{ dsfont,subfig,subfloat,color }
\newcommand{\ud}{\,\mathrm{d}}
\usepackage{caption}

\DeclareCaptionType{mytype}[Flowchart][List of mytype]

\begin{document}
\IEEEoverridecommandlockouts


\title{Continuous-Domain Solutions of Linear Inverse Problems with Tikhonov vs. Generalized TV Regularization }

\author{Harshit~Gupta,~
        Julien~Fageot,~
        and~Michael~Unser~
\thanks{The   authors   are   with   the   Biomedical   Imaging   Group,   \'{E}cole   polytechnique   f\'{e}d\'{e}rale   de
Lausanne,   Lausanne   1015,   Switzerland. This project has been funded by H2020-ERC, Grant agreement No. 692726-GlobalBioIm.}}



\maketitle

\begin{abstract}
We consider linear inverse problems that are formulated in the continuous domain. The object of recovery is a function that is assumed to minimize a convex objective functional. The solutions are constrained by imposing a continuous-domain regularization. We derive the parametric form of the solution (representer theorems) for Tikhonov (quadratic) and generalized total-variation (gTV) regularizations. We show that, in both cases, the solutions are splines that are intimately related to the regularization operator. In the Tikhonov case, the solution is smooth and constrained to live in a fixed subspace that depends on the measurement operator. By contrast, the gTV regularization results in a sparse solution composed of only a few dictionary elements that are upper-bounded by the number of measurements and  independent of the measurement operator. Our findings for the gTV regularization resonates with the minimization of the $\ell_1$ norm, which is its discrete counterpart and also produces sparse solutions. Finally, we find the experimental solutions for some measurement models in one dimension. We discuss the special case when the gTV regularization results in multiple solutions and devise an algorithm to find an extreme point of the solution set which is guaranteed to be sparse.
\end{abstract}
\section{Introduction}
In a linear inverse problem, the task is to recover an unknown signal from a finite set of noisy linear measurements. To solve it, one needs a forward model that describes how these measurements are acquired. Generally, this model is stated as the continuous-domain transform of a continuous-domain signal. For example, MRI data is modeled as the samples of the Fourier transform of a continuous-domain signal. The traditional approach to state this inverse problem is to choose an arbitrary but suitable basis $\{\varphi_n\}$ and to write that the reconstructed signal is
\begin{equation}\label{basis}
f(x)=\sum_{n=1}^N \text{f}_n \varphi_n(x),
\end{equation}
where $\M f=(\text{f}_1,\ldots,\text{f}_N)\in \R^N $. Given the measurements $\V z \in \R^M$, the task then is to find the expansion coefficients $\M f$ by minimizing 
\begin{equation}\label{discrete2}
\M f^*=\arg \min_{\M f \in \R^N}\left( \underbrace{\|\V z- \M H \M f\|_2^2}_{\text{I}}+\lambda\underbrace{ \|\M L \M f \|^2_{2}}_{\text{II}}\right),
\end{equation} 
where $\M H: \R^M\times \R^N$ has elements $[\M H]_{m,n}=\langle h_m, \varphi_n \rangle$. The analysis functions $\{h_m\}_{m=1}^M$ specify the forward model which encodes the physics of the measurement process. Term $\rm{I}$ in \eqref{discrete2} is the data fidelity. It ensures that the recovered signal is close to the measurements. Term $\rm{II}$ is the regularization, which encodes the prior knowledge about the signal. The regularization is imposed on some transformed version of the signal coefficients using the matrix $\M L$. Various linear \cite{Tikhonov1963,Bertero1998} and iterative algorithms \cite{Figueiredo2003},\cite{Lustig2007},\cite{Figueiredo2007} have been developed to solve Problem \eqref{discrete2}. In recent years, the notion that the real-world signals are sparse in some basis (e.g., wavelets) has become popular.
This prior is imposed by using the sparsity-promoting $\ell_1$-regularization norm \cite{Donoho2006},\cite{Candes2007} and results in the minimization problem 
 \begin{equation}\label{discrete1}
\M f^*=\arg \min_{\M f \in \R^N} \left( \|\V z- \M H \M f\|_2^2+ \lambda\|\M L \M f \|_{1} \right).
\end{equation} 
The solutions to \eqref{discrete2}, \eqref{discrete1}, and their variants with generalized data-fidelity terms are well known~\cite{arthur1970},\cite{Tibshirani1996},\cite{efron2007},\cite{Unser2016b}.\\
While those discretization paradigms are well studied and used successfully in practice, it remains that the use of a prescribed basis $\{\varphi_n\}$, as in \eqref{basis}, is  somewhat arbitrary.

In this paper, we propose to bypass this limitation by reformulating and solving the linear inverse problem directly in the continuous domain. To that end, we impose the regularization in the continuous domain, too, and restate the reconstruction task as a functional minimization. We show that this new formulation leads to the identification of a natural basis for the solution; this results in an exact discretization of the problem.

\noindent Our contributions are summarized as follows:
\begin{itemize}
\item 
Given $\V z \in \R^{M} $, we formalize the inverse problem in the continuous domain as
\begin{equation}\label{into:cd}
f_R=\arg \min_{f \in \Spc X}\underbrace{ \left(\|\V z- \Op H\{f\}\|_2^2+\lambda R(f)\right)}_{J_{R}(\V z|f)} ,
\end{equation}
where $f$ is a function that belongs to a suitable function space $\Spc X$.
Similarly to the discrete regularization terms $\|\M L \M f \|^2_{\ell_2}$ and $\|\M L \M f \|_{\ell_1}$ in \eqref{discrete2} and \eqref{discrete1}, we focus on their continuous-domain counterparts $R(f)=\|\Lop f\|^{2}_{L_2}$ and $R(f)=\|\Lop f\|_{\Spc M}$, respectively. There, $\Lop$ and $\Op H$ are the continuous-domain versions of $\M L $ and $\M H$, while $\|\Lop f\|_{\Spc M}$ is the proper continuous-domain counterpart of the discrete $\ell_1$ norm. We show that the effect of these regularizations is similar to the effect of their discrete counterparts. 

\item We provide the parametric form of the solution (representer theorem) that minimizes $J_R(\V z|f)$ in \eqref{into:cd} for the Tikhonov regularization $R(f)=\|\Lop f\|^2_{L_2}$ and the generalized total-variation (gTV) regularization $R(f)=\|\Lop f\|_{\Spc M}$. Our results underline how the discrete regularization resonates with the continuous-domain one. The optimal solution for the Tikhonov case is smooth, while it is sparse for the gTV case. The optimal bases in the two cases are intimately connected to the operators $\Lop$ and $\Op H$.
\item We present theoretical results that are valid for any convex and lower-semicontinuous data-fidelity term. This includes the case when the data-fidelity term is $\|\V z- \Op H\{f\}\|_2^2$.
\item We propose an exact discretization scheme to minimize $J_R(\V z|f)$ in the continuous domain. Even though the minimization of $J_{R}(\V z|f)$ is an infinite-dimensional problem, the knowledge of the optimal basis of the solution makes the problem finite-dimensional: it boils down to the search for a set of optimal expansion coefficients.

\item We devise an algorithm to find a sparse solution when the gTV solution is non-unique. For this case, the optimization problem turns out to be a LASSO \cite{Tibshirani1996} minimization with non-unique solution. We introduce a combination of FISTA \cite{beck2009} and the simplex algorithm to find a sparse solution which we prove to be an extreme point of the solution set.
 \end{itemize}
 
\noindent  The paper is organized as follows: In Sections 2 and 3, we present the formulation and the theoretical results of the inverse problem for the two regularization cases. In Section~4, we compare the solutions of the two cases. We present our numerical algorithm in Section 5 and illustrate its behavior with various examples in Section 6. The mathematical proofs of the main theorems are given in the appendices and the supplementary material. 
\subsection{Related Work}
The use of $R(f)=\|\Lop f\|^2_{L_2}$ goes back to Tikhonov's theory of regularization \cite{Tikhonov1963} and to kernel methods in machine learning \cite{Scholkopf:2001:LKS:559923}. In the learning community, representer theorems (RT) as in \cite{scholkopf2001generalized},\cite{Wahba1990} use the theory of reproducing-kernel Hilbert spaces (RKHS) to state the solution of the problem for the restricted case where the measurements are samples of the function. For the generalized-measurement case, there are also tight connections between these techniques and variational splines and radial-basis functions \cite{wahba1999support},\cite{bezhaev2001variational}, \cite{wendland2004scattered}. These representer theorems, however, either have restrictions on the empirical risk functional or on the class of measurement operators.
 
Specific spline-based methods with quadratic regularization have been developed for inverse problems. In particular, \cite{kybic2002generalized}, \cite{Kybic2002} used variational calculus. Here, we strengthen these results by proving the unicity and existence of the solution of \eqref{into:cd} for $R(f)=\|\Lop f\|^2_{L_2}$. We revisit the derivation of the result using the theory of RKHS.

Among more recent non-quadratic techniques, the most popular ones rely on (TV) regularization which was introduced as a noise-removal technique in \cite{Rudin1992} and is widely used in computational imaging and compressed sensing, although always in discrete settings. Splines as solutions of TV problems for restricted scenarios have been discussed in \cite{Steidl2006}. More recently, a RT for the continuous-domain $R(f) = \lVert \Op L f\rVert_{\mathcal{M}}$ in a general setting has been established in \cite{Unser2016}, extending the seminal work of Fisher and Jerome  \cite{Fisher1975}.  The solution has been shown to be composed of splines that are directly linked to the differential operator $\Lop$. Other recent contributions on inverse problems in the space of measures include \cite{Bredies2013inverse,Candes2013super,Denoyelle2015support,Chambolle2016geometric, Flinth2017exact}. In particular, in this paper, we extend the result of \cite{Unser2016} to an unconstrained version of the problem. 

\section{Formulation}
In our formulation of a linear inverse problem, the signal $f$ is a function of the continuous-domain variable $ x \in \Rd$. The task is then to recover $f $ from the vector of measurements $\V z=\Op H\{f\}+\M n \in \Rm$, where $\M n$ is an unknown noise component that is typically assumed to be i.i.d. Gaussian. \\
In the customary discrete formulation, the basis of the recovered function is already chosen and, therefore, all that remains is to recover the expansion coefficients of the signal representation \eqref{basis}. In this scenario, one often includes matrices $\M H$ and $\M L$ that directly operate on these coefficients. However, for our continuous-domain formulation, the operations have to act directly on the function $f$. For this reason, we also need the continuous-domain counterparts of the measurement and regularization operators.
The entities that enter our formulation are described next.

\subsection{Measurement Operator} 
The system matrix $\M H$ in \eqref{discrete2}  and \eqref{discrete1} is henceforth replaced by the operator $\Op H: \Spc X \to \R^M $ that maps the continuous-domain functions living in the space $\Spc X$ to the linear measurements $\V z \in \R^N$. This operator is described as 
\begin{equation}
\Op H\{f\}= \left( \langle h_1,f \rangle, \ldots, \langle h_M, f \rangle\right)=\left( z_1,\ldots,z_M\right)=\V z,
\end{equation}
where $\langle h,g \rangle=\int_{\R} h(x) g(x) \ud x$.
For example, the components of the measurement operator that samples a function at the locations $x_1,\ldots,x_M$ are modeled by $h_m=\delta(\cdot-x_m)$. Similarly, Fourier measurements at pulsations $\omega_1,\ldots,\omega_M$ are obtained by taking $h_m={\rm{e}}^{-{\rm{j}} \omega_m (\cdot)}$.
\subsection{Data-Fidelity Term}
\noindent  As extension of the conventional quadratic data-fidelity term $\|\V z- \M H \M f\|_2^2$, we consider the general convex cost functional 
$E : \R^M \times \R^M \to \R^+ \cup \{\infty\}$ that measures the discrepancy between the measurements $\M z$ and the values $\Op H\{f\}$ predicted from the reconstruction. A relevant example is the Kullback-Leibler (KL)-divergence, which is often used as the data-fidelity term when the measurements are corrupted by Poisson noise \cite{Csiszar1991}. Alternatively, when the measurements are noiseless, we use the indicator function 
\begin{align}\Spc I(\V z_0, \Op H \{f\})=\begin{cases}  \phantom{0}0,\quad \V z_0=\Op H \{f\}\\ \infty, \quad \V z_0\neq\Op H \{f\},   \end{cases} \end{align}
which imposes an exact fit.
We assume that $E$ is a convex lower semi-continuous function with respect to its arguments. This will enable us to state the existence of a solution and use convex optimization techniques to find the minimum of the objective functional.
\subsection{ Regularization Operator}
\noindent Since the underlying signal is continuously defined, we need to replace the regularization matrix $\M L$ in \eqref{discrete2} and \eqref{discrete1} by a regularization operator $\Op L:\Spc X \to \Spc Y$, where $\Spc X$ and $\Spc Y$ are appropriate function spaces to be defined in Section \ref{spaces}. The typical example that we have in mind is the derivative operator $\Op L=\Op D=\frac{\rm{d \phantom{x}} }{{\rm{d}} x}$. 
The continuous-domain regularization is then imposed on $\Lop f$. We assume that the operator $\Lop$ is admissible in the sense of defintion \ref{def:splineadmiss}.
\begin{definition}\label{def:splineadmiss}
The operator $\Lop: \Spc X \to \Spc Y$ is called \emph{\textbf{spline-admissible}} if
\begin{itemize}
\item it is linear and shift-invariant;
\item its null space $ \Spc N_{\Lop}=\{p \in \Spc X : \Lop p=0\}$ is finite-dimensional; 
\item it admits the Green's function $\rho_{\Lop}: \R \to \R$ with the property that $\Lop\rho_{\Lop}=\delta$.
\end{itemize}
\end{definition}
Given that $\widehat{\Lop}$ is the frequency response of $\Lop$, the Green's function can be calculated through the inverse Fourier transform $\rho_{\Lop}=\mathcal{F}^{-1} \left\lbrace \frac{1}{\displaystyle{\widehat{\Lop} } }\right\rbrace $. For example, if $\Lop=\Op D$, then $\rho_{\Op D}(x)=\frac{1}{2} \text{sign}(x)$.

\subsection{ Regularization Norms}
\noindent Since the optimization is done in the continuous domain, we also have to specify the proper counterparts of the $\ell_2$ and $\ell_1$ norms, as well as the corresponding vector spaces. 
\renewcommand{\theenumi}{\roman{enumi}}
\begin{enumerate}
\item Quadratic (or Tikhonov) regularization: $R_{\rm{Tik}}(f)=\|\Lop f\|^2_{L_2}$, where
\begin{equation}
\|w\|_{L_2}^2:= \int_{\R}|w( x)|^2\ud x.
\end{equation}
\item Generalized total variation: $R_{\rm{gTV}}(f)=\|\Lop f\|_{\Spc M}$, where
\begin{align}
\|w\|_{\Spc M}:=& \sup_{\varphi \in \Spc C_0(\R), \|\M \varphi\|_{\infty}=1}\langle{ w, \mathbb{\varphi}}\rangle. \label{M intro}
\end{align}
There, $\Spc C_0(\R)$ is the space of continuous functions that decay to 0 at infinity.
Moreover, $\Spc M=\{w:\R \rightarrow \R ~|~ \|w\|_{\Spc M} < \infty\}$.
In particular, when $w \in  L_1\subset \Spc M$, we have that
\begin{align}
\|w\|_{\Spc M}=& \int_{\R}|w( x)|\ud x =\|w\|_{L_1}.
\end{align}
Yet, we note that $\Spc M$ is slightly larger than $L_1$  since it also includes the Dirac distribution  $\delta$ with $\|\delta\|_{\Spc M}=1$.
The popular TV norm is recovered by taking $ \|f\|_{\text{TV}}=\|\Op D f\|_{\Spc M}$ \cite{Unser2016}.
\end{enumerate}
\subsection{Search Space}\label{spaces}
\noindent The Euclidean search space $\R^N $ is replaced by spaces of functions, namely,
 \begin{align}
 \bl=&\{f : \R \rightarrow \R \ | \  \|\Lop f\|_{L_2}<+\infty\},\\
\bltv=&\{f : \R \rightarrow \R \ | \  \|\Lop f\|_{\Spc M}<+\infty\}.
\end{align}
In other words, our search (or native) space is the largest space over which the regularization is well defined. It turns out that $\bl$ and $\bltv$ are Hilbert and Banach spaces, respectively. This means that there exists a well defined inner product  $\langle\cdot,\cdot\rangle_{\bl}$ on  $\bl$ and a norm $\|\cdot\|_{\bltv}$ on $\bltv$.
The structure of these spaces has been studied in \cite{Unser2016} and is recalled in the supplementary material.

\noindent As we shall in Section \ref{sec:Theoretical Results}, the solution of \eqref{into:cd} will be composed of splines; therefore, we also review the definition of the splines.
\begin{definition}[\textbf{Nonuniform $\Lop$-spline}]
A function $f: \R \to \R$ is called a nonuniform $\Lop$-spline with spline knots ($ x_1,\ldots, x_K$) and weights ($a_1,\ldots,a_K $) if 
\begin{align}\label{def:non-uniformSpline}
\Lop f=\sum_{k=1}^K a_k \delta( \cdot - x_k ).
\end{align}
\end{definition}
\noindent By solving the differential equation in \eqref{def:non-uniformSpline}, we find that the generic form of the nonuniform spline $f$ is
\begin{align}
f=p_0+\sum_{k=1}^K a_k \rho_{\Lop}(\cdot-x_k),
\end{align}
where $p_0 \in \nl$. Note that $ \rho_{\Lop}(\cdot-x_k)=\Lop^{-1}\{\delta(\cdot-x_k)\}$, where $\Lop^{-1} : f \mapsto \Lop^{-1}f=\rho_{\Lop}*f$, is the shift-invariant inverse of $\Lop$.

\section{Theoretical Results}\label{sec:Theoretical Results}
To state our theorems, we need some technical assumptions. 
\begin{assumption} 
\renewcommand{\theenumi}{\roman{enumi}}
\begin{enumerate} 
\item The bounded vector-valued functional $\Op H:\Spc X \rightarrow \Rm  $ gives the linear measurements $f\mapsto\Op H \{f\}= (\langle h_1,f\rangle,\ldots,\langle h_M,f\rangle)$. 
\item The functional $E : (\R^M \times \R^M) \to \R^+  \cup \{\infty\}$ is convex and lower semi-continuous.
\item The regularization operator  $\Lop: \Spc X \to \Spc Y$ is \emph{spline-admissible}. Its finite-dimensional null space $\nl$ has the basis $\V p=(p_1,\ldots ,p_{N_0})$.
\item The inverse problem is well posed over the null space. This means that, for any pair $p_1, p_2 \in \nl$, we have that
\begin{align}
\Op H \{p_1\}=\Op H \{p_2\}\Leftrightarrow p_1=p_2.
\end{align}
In other words, different null-space functions result in different measurements.
\end{enumerate}
\end{assumption}
In particular Condition iv) implies that $\nl \cap \Spc N_{\Op H}= \{0\}$, where $\Spc N_{\Op H}$ is the null space of the vector-valued measurement functional. This property is essential to make the optimization problem \eqref{into:cd}  well posed. This kind of requirement is common to every regularization scheme.

We now state our two main results. Their proofs are given in Appendix \ref{apppx:L2} and Appendix \ref{app:proofL1theo}.
\subsection{Inverse Problem with Tikhonov/$L_2$ Regularization }
\begin{theorem}\label{theo:L2_representer}
Let Assumption 1 hold for the search space $\Spc X= {\bl}$ and regularization space $\Spc Y=  L_2$. Then, the minimizer
\begin{align}\label{empirical} 
 f_2=\arg\min_{f \in \bl }\left( E(\V z, \Op H(f))+ \lambda \| \Lop f\|^2_{L_2} \right)
\end{align}
 is unique and admits a parametric solution of the form
  \begin{equation}\label{eqn:spline_L2}
  f_2(x)=\sum_{m=1}^{M} a_m \varphi_m(x)+\sum_{n=1}^{N_0} b_n p_n(x),
  \end{equation}
  where $\varphi_m=\mathcal{F}^{-1}\left\lbrace\frac{\widehat{h}_m}{|\widehat{\Lop}|^2}\right\rbrace=(\Lop^*\Lop )^{-1}h_m$, $\M a=(a_1,\ldots,a_M)$, and $\M b=(b_1,\ldots,b_{N_0})$ are expansion coefficients such that 
\begin{align}\label{eqn:ortho}
\sum_{m=1}^M a_m \langle h_m,p_n\rangle=0
\end{align} for all $n \in \{1,\ldots, N_0\}$.
\end{theorem}

\subsection{Inverse Problem with gTV Regularization}


\begin{theorem}	\label{theo:L1_representer}
Let Assumption 1 hold for the search space $\Spc X= { \bltv}$ and regularization space $\Spc Y= \Spc M$. Moreover, assume that $\Op H$ is weak*-continuous (see Supplementary Material). Then, the set 
\begin{align}\label{empiricalL1} 
\Spc V= \left\{\arg \min_{f \in \Spc X_{1} }\left( E(\V z, \Op Hf)+ \lambda \| \Lop f\|_{\Spc M} \right)\right\}
\end{align}
of minimizer
is nonempty, convex, weak*-compact, and its extreme points are nonuniform $\Lop$-splines of the  form
\begin{align} \label{eq:L1formsolutions}
f_1(x)=\sum_{k=1}^K a_k \rho_\Lop(x-x_k)+\sum_{n=1}^{N_0} b_n p_n(x)
\end{align}
for some $K\leq (M-N_0)$. The parameters of the solution are the unknown knots  $(x_1,\ldots,x_K)$ and the expansion coefficients $\M a=(a_1,\ldots,a_K), \M b=(b_1,\ldots,b_{N_0})$. The solution set $\Spc V$ is the convex hull of these extreme points and $\|\Lop f\|_{\Spc M}=\|\M a\|_{1}$. 
\end{theorem}

The existence and nature of the solution set in these 2 cases is stated jointly in Theorem \ref{solution set}. The proof is given in Appendix \ref{appx:solutionset}.
\begin{theorem}\label{solution set}
Let Assumption 1 hold where $\Spc X$ is the search space and $\Spc Y$ is the regularization space. Then, every member of the solution set
\begin{align}\label{empirical} 
 \Spc V=\left\{\arg\min_{f \in \Spc X }\left( E(\V z, \Op H \{f\})+ \lambda R(f)  \right)\right\},
\end{align}
where $R$ is either $R_{\rm{Tik}}$ or $R_{\rm{gTV}}$, has the same measurement $\V z_{0}$ given that the problem has at least one solution.
\end{theorem}
Theorem \ref{solution set} implies that, for the gTV case when $E$ is strongly convex, the elements of the solution set $\Spc V$ map to the unique point $\V z_{\Spc V}=\Op H\{ f\}, \, \forall \, f \in \Spc V$. 
\subsection{Illustration with Ideal Sampling}

\noindent Here, we discuss the regularized case where noisy data points $\left((x_1,z_1),\ldots,(x_M,z_M)\right)$ are fitted by a function. The measurement functionals in this case are the shifted Dirac impulses $h_{m}=\delta(\cdot-x_m)$ whose Fourier transform is $\widehat{h}_{m}(\omega)={\rm{e}}^{-{\rm{j}} \omega x_m}$. We choose $\Lop=\Op D$ and $E= \|\V z- \Op H\{f\}\|_2^2$. 
For the $L_2$ problem, we have that
\begin{align}\label{l2_example_problem}
f_{2}=\arg \min_{f \in \Spc X_2} \left(\sum_{m=1}^{M}|z_m-f(x_m)|^2+\lambda \|\Op D f\|^2_{L_2} \right).
\end{align}
As given in Theorem \ref{theo:L2_representer}, $f_2$ is unique and has the basis function $\varphi_{m}(x)=\mathcal{F}^{-1}\left\lbrace\frac{{{\rm{e}}^{-{\rm{j}} (\cdot) x_m}}}{|\rm{j} (\cdot)|^2}\right\rbrace(x)=\frac{1}{2}|x-x_m|$. The resulting solution is piecewise linear. It can be expressed as 
\begin{align}\label{l2_example_solution}
f_{2}(x)= b_1+\sum_{m=1}^{M}  \frac{1}{2} a_m |x-x_m| ,
\end{align} where $b_1 \in \Spc N_{\Op D}$ is a constant.\\
 We contrast \eqref{l2_example_problem} with the gTV version
 \begin{align}\label{l1_example_problem}
f_{1}=\arg \min_{f \in \bltv} \left(\sum_{m=1}^{M}|z_m-f(x_m)|^2+\lambda \underbrace{\|\Op D f\|_{\Spc M}}_{\| f\|_{\text {TV}}} \right).
\end{align}
In this scenario, the term $\|\Op Df\|_{\Spc M}$ is the total variation of the function $f$. It penalizes solutions that vary too much from one point to the next.

One readily checks that $\rho_{\Op D}=\mathds 1_{+}$ is a Green's function of $\Op D$ since it satisfies $\Op D\{\mathds 1_{+}\}=\delta$. Based on Theorem \ref{theo:L1_representer}, any extreme point of  (\ref{l1_example_problem}) is of the form
\begin{align}\label{l1_example_solution}
f_{1}(x)=b_1+ \sum_{k=1}^{K} a'_k {\mathds 1}_+(x-\tau_k),
\end{align} 
which is a piecewise constant function composed of a constant term $b_1$ and  $K \leq( M-1)$ unit steps (Heaviside functions) located at $\{\tau_{k}\}^K_{k=1}$. These knots are not fixed a priori and usually differ from the measurement points $ \{x_m\}^M_{m=1}$.

The two solutions and their basis functions are illustrated in Figure \ref{demo} for specific data.
This example demonstrates that the mere replacement of the $L_2$ penalty with the gTV norm has a fundamental effect on the solution: piecewise-linear functions having knots at the sampling locations are replaced by  piecewise-constant functions with a lesser number of adaptive knots. 
Moreover, in the gTV case, the regularization has been imposed on the derivative of the function $\left( \|\Op D f\|_{\Spc M}\right)$, which uncovers the innovations $\Op D  f_{1}=\sum_{k=1}^{K} a'_k \delta(\cdot-\tau_k)$. 
By contrast, when $R(f)=\|\Op D f\|^2_{L_2}=\langle \Op D^*\Op Df,f \rangle$, the recovered solution is such that $\Op D^*\Op D f_{2}=\sum_{m=1}^{M} a_m \delta(\cdot-x_m)$, where $\Op D^*=-\Op D$ is the adjoint operator of $\Op D$. 
Thus, in both cases, the recovered functions are composed of the Green's function of the corresponding active operators: $\Op D $ vs. $\Op D^*\Op D=-\Op D^2$.

\begin{figure}
\begin{minipage}[b]{0.96\linewidth}
  \centering
  \centerline{\includegraphics[width=0.95\linewidth,height=4cm, trim={0 0 0 7mm},clip=true ]{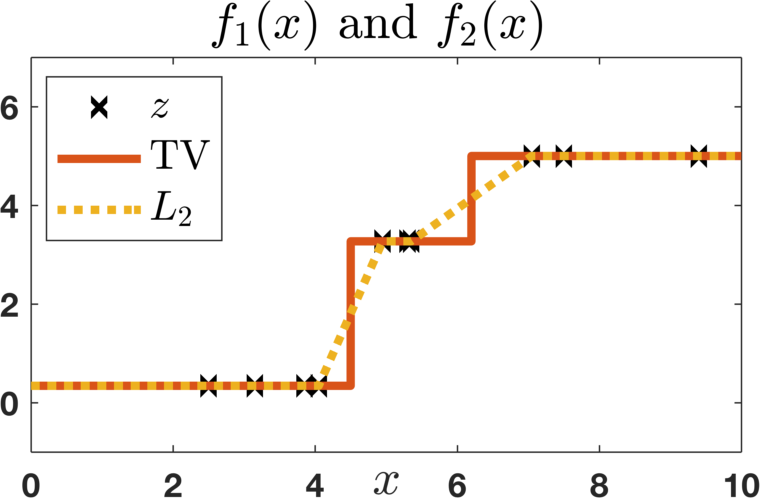}}
  \centerline{(a) $f_1(x)$ and $f_2(x)$.}
  \vspace{0.1cm}
   \includegraphics[width=\linewidth,height=3.5cm, trim={0 0 0 0},clip=true]{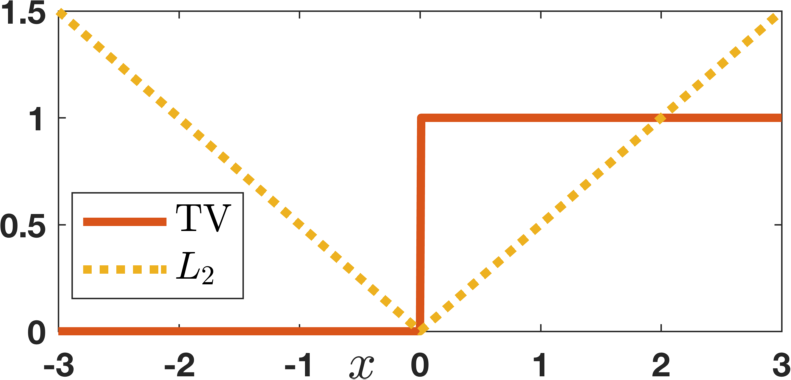}
  \centerline{(b) $\rho_{\Op D}(x)$ and $\rho_{\Op D^*\Op D}(x)$.}
\end{minipage}
\caption{Reconstructions of a signal from nonuniform samples for $\Lop =\Op D$: (a) Tikhonov ($L_2$) vs. gTV solution, and (b) Corresponding basis functions $\rho_{\Op D}$ vs. $\rho_{\Op D^*\Op D}$. Note that the gTV solution is non-unique since, for example, any nondecreasing piecewise-constant interpolation between the fourth and the fifth measurement has the same arc-length as the solution shown.
}\label{demo}
\end{figure}
\section{Comparison}
We now discuss and contrast the results of Theorems \ref{theo:L2_representer} and \ref{theo:L1_representer}.  In either case, the solution is composed of a primary component and a null-space component whose regularization cost vanishes. 

\subsection{Nature of the Primary Component}
\subsubsection{Shape and Dependence on Measurement Functionals} The solution for the gTV regularization is composed of atoms within the infinitely large dictionary $\{\rho_{\Lop}(\cdot-\tau)\},\, \forall \tau \in \R$, whose shapes depend only on $\Lop$. In contrast, the $L_2$ solution is composed of fixed atoms $\{\varphi_m\}_{m=1}^M$ whose shapes depend on both $\Lop $ and $\Op H$. As the shape of the atoms of the gTV solution does not depend on $\Op H$, this makes it easier to inject prior knowledge in that case.
\subsubsection{Adaptivity} The weights and the location of the atoms of the gTV solution are adaptive and found through a data-dependent procedure which results in a sparse solution that turns out to be a nonuniform spline. By contrast, the $L_2$ solution lives in a fixed finite-dimensional space.  

\subsection{Null-Space Component}  The second component in either solution belongs to the null space of the operator $\Lop$. As its contribution to regularization vanishes, the solutions tend to have large null-space components in both instances.

\subsection{Oscillations} The modulus of the Fourier transform of the basis function of the gTV case, $\left\lvert\left\{\frac{1}{\widehat{\Lop}}\right\}\right\rvert$  typically decays faster than that of the $L_2$ case, $\left\lvert \left\{\frac{\widehat{h}_m}{|\widehat{\Lop}|^2}\right\}\right\lvert$. Therefore, the gTV solution exhibits weaker Gibbs oscillations at edges.

\subsection{Unicity of the Solution}
Our hypotheses guarantee existence. Moreover, the minimizer of the $L_2$ problem is unique. By contrast, the gTV problem can have infinitely many solutions, despite all having the same measurements. The solution set in this case is convex and the extreme points are  nonuniform splines with fewer knots than the number ($M-N_0$) of measurements. When the gTV solution is unique, it is guaranteed to be an $\Lop $-spline.
\subsection{Nature of the Regularized Function} One of the main differences between the reconstructions $f_2$ and $f_1$ is their sparsity. Indeed, $\Lop f_{1}$ uncovers Dirac impulses situated at $(M-1)$ locations for the gTV case, with $\Lop f_{1}=\sum_{m=1}^{M-1} a_m \delta(\cdot -\tau_m)$. In return, $\Lop f_{2}$ is a nonuniform $\Lop$-spline convolved with the measurement functions, whose temporal support is not localized. This allows us to say that the gTV solution is sparser than the Tikhonov solution.

\section{Discretization and Algorithms}
\label{sec:algorithm}
 We now lay down the discretization procedure that translates the continuous-domain optimization into a more tractable finite-dimensional problem. 
Theorems \ref{theo:L2_representer} and \ref{theo:L1_representer} imply that the infinite-dimensional solution lives in a finite-dimensional space that is characterized by the basis functions $\{\varphi_m\}^M_{m=1}$ for $L_2$ and $\{\rho_{\Lop}(\cdot-\tau_k)\}_{k=1}^K$ for gTV, in addition to $\{p_n\}_{n=1}^{N_0}$ as basis of the null space. 
Therefore, the solutions can be uniquely expressed with respect to the finite-dimensional parameter $\M a \in \R^M$ or $\M a \in \R^K$, respectively, and $\M b \in \R^{N_0}$.
Thus, the objective functional $J_{R}(\V z|\lambda,f)$ can be discretized to get the objective functional $J_{R}(\V z|\lambda,\M a,\M b)$. Its minimization is done numerically, by expressing $\Op H\{f\}$ and $\|\Lop f\|^2_{L_2}$ or $\|\Lop f\|_{\Spc M}$ in terms of $\M a$ and $\M b$. We discuss the strategy to achieve $J_{R}(\V z|\lambda,\M a,\M b)$ and its minima for the two cases.
\subsection{Tikhonov Regularization}
For the $L_2$ regularization, given $\lambda >0$, the solution 
\begin{align}\label{l2general} 
 f_2=\arg\min_{f \in \bl }\underbrace{\left( E(\V z, \Op H\{f\})+ \lambda \| \Lop f\|^2_{L_2} \right)}_{J_{2}(\V z|\lambda,f)}
\end{align}
can be expressed as
 \begin{equation}\label{l2generalsol}
  f_2=\sum_{m=1}^{M} a_m \varphi_m+\sum_{n=1}^{N_0} b_n p_n.
  \end{equation}
  Recall that $\varphi_m=(\Lop^*\Lop)^{-1}h_m$, so that
  \begin{equation}\label{l*l2generalsol}
 \Lop^*\Lop f_2=\sum_{m=1}^{M} a_m h_m.
  \end{equation}
The corresponding $J_{2}(\V z|\lambda, \M a,\M b)$ is then found by expressing $\Op H \{f_2\}$ and $ \| \Lop f_2\|^2_{L_2}$ in terms of $\M a$ and $\M b$.
Due to the linearity of the model,
 \begin{align}
    \Op H\{f_2\}&= \sum_{m=1}^{M} a_m \Op H \{\varphi_m\}+\sum_{n=1}^{N_0} b_n \Op H \{p_n\}\nonumber \\
    &= \M V \M a+\M W \M b,
 \end{align}
where $ [\M V]_{m,n}= \langle {h_m},{\varphi_n}\rangle $ and $[\M W]_{m,n}=\langle{h_m},{p_n}\rangle$. Similarly,
 \begin{align}
   \langle \Op Lf_2,\Op Lf_2\rangle&=   \langle \Op L^* \Op L f_2,f_2\rangle=\left\langle \sum_{m=1}^{M} a_m h_m, f_2\right\rangle \label{firststep}\\
    &= \M a^T \M V \M a +\M a^T \M W \M b=\M a^T \M V \M a\label{secondstep},
  \end{align}
  where \eqref{firststep} uses \eqref{l*l2generalsol} and where \eqref{secondstep} uses the orthogonality property \eqref{eqn:ortho}, which we can restate as $\M a^T \M W=\M 0$.
By substituting these reduced forms in \eqref{l2general}, the discretized problem becomes
  \begin{equation}
 f_2=\arg \min_{\M a, \M b}  \underbrace{\left(E(\V z,  \M V \M a+\M W \M b)+\lambda \M a^T \M V \M a\right)}_{J_{2}(\V z|\lambda, \M a,\M b)=J_{2}(\V z|\lambda,f_2)}.
\end{equation} 
  Due to Assumption 1.ii),  this problem is convex. If $E$ is differentiable with respect to the parameters, the solution can be found by gradient descent.\\
When $E(\V z, \Op H \{f\})=\|\V z -\Op H \{f\}\|_2^2$, the problem is reduced to
 \begin{equation}
  \arg \min_{\M a, \M b} \underbrace{\left( \|\V z- (\M V \M a+\M W \M b)\|_2^2+\lambda \M a^T \M V \M a \right)}_{J_{2}(\V z|\lambda, \M a,\M b)}
\end{equation} 
 which is very similar to \eqref{discrete2}.
 This criterion is convex with respect to the coefficients $\M a$ and $ \M b$. 
 Enforcing that the gradient of $J_2$ vanishes with respect to $\M a$ and $\M b$ and setting the gradient to $\M 0$ then yields $M$ linear equations with respect to the $M+N_0$ variables, while the orthogonality property \eqref{eqn:ortho} gives $N_0$ additional constraints. 
  The combined equations correspond to the linear system 
\begin{align}\label{eqn:L2matrix}
\begin{bmatrix}
\M V+\lambda \M I& \M W \\
\M W^T & \M 0\\
\end{bmatrix}
\begin{bmatrix}
\M a\\
\M b
\end{bmatrix}&=\begin{bmatrix}
\M z\\
\M 0
\end{bmatrix}.
\end{align}
The system matrix so obtained can be proven to be positive definite due to the property of Gram matrices generated in an RKHS and the admissibility condition of the measurement functional (Assumption 1). 
This ensures that the matrix is always invertible. 
The consequence is that the reconstructed signal can be obtained by solving a linear system of equation, for instance by QR decomposition or by simple matrix inversion. The derived solution is the same as the least-square solution in \cite{Kybic2002}.
\subsection{gTV Regularization} 
{
\noindent In the case of gTV regularization, the problem to solve is 
\begin{align}\label{l1general} 
 f_1=\arg\min_{f \in \bl }\underbrace{\left( E(\V z, \Op H \{f\})+ \lambda \| \Lop f\|_{\Spc M} \right)}_{J_{1}(\V z|\lambda,f)}.
\end{align}
According to Theorem \ref{theo:L1_representer}, an extreme-point solution of \eqref{l1general} is
\begin{align}
f_{1}(x)= \sum_{k=1}^{K} a_{k} \rho_{\Lop}(x-\tau_k)+\sum_{n=1}^{N_0}b_n p_n(x)
\end{align}and satisfies
\begin{equation}
	\Op L f_1 = w_1 = \sum_{k=1}^K a_k \delta( \cdot - \tau_k)
\end{equation}
with $K \leq (M-N_0)$. Theorem \ref{theo:L1_representer} implies that we only have to recover $a_k$, $\tau_k$, and the null-space component $p$ to recover $f_1$.

Since we usually know neither $K$ nor
$\tau_k$ beforehand, our solution is to quantize the $x$-axis and look for $\tau_k $ in the range $[0,T]$ on a grid with $N \gg K $ points. We control the quantization error with the grid step $\Delta=T/N$.  
The discretized problem is then to find $\M a \in \R^N$ with fewer than $(M - N_0)$ nonzero coefficients and $\M b \in \R^{N_0}$ such that 
 \begin{equation}
 f_{1,\Delta }(x) = \sum_{n=0}^{N-1} a_{n} \rho_{\Op L} (x - n \Delta ) + \sum_{n=1}^{N_0} b_n p_n(x)
 \end{equation}
satisfies \eqref{l1general}, with $K \leq (M-N_0) \ll N$ nonzero coefficients $a_n$.
When the discretization step $\Delta$ goes to $0$ (or when $N$ is large enough), we recover the solution of the original problem \eqref{l1general}.}

Similarly to the $L_2 $ case, $J_{1}(\V z|\lambda, \M a,\M b)$ is found by expressing $ \Op H \{f_{1,\Delta}\}$ and $ \| \Lop f_{1,\Delta}\|_{\Spc M}$ in terms of $\M a$ and $\M b$. 
 For this, we use the properties that $\Lop\rho_{\Lop}=\delta$, $\|\delta\|_{\rm TV}=1$, and $\Lop p_n=0$ for $n \in [1\ldots N_0]$. 
 This results in
\begin{align}\label{eqn:exact1}
\Op H \{f_{1,\Delta}\}&=\M P \M a +\M Q \M b ,\\
 \| \Lop f_{1,\Delta}\|_{\Spc M }&= \|\M a\|_1,\label{eqn:exact2}
\end{align}
 where $\M a=(a_0,\ldots,a_{N-1}),~[\M P]_{m,n}=\langle h_m,\rho_{\Lop}(\cdot-n \Delta)\rangle $ for $ n\in [0\ldots N -1]$, $[\M Q]_{m,n}=\langle h_m,p_n\rangle$ for $n\in [1 \ldots N_0] $, $\|\M a\|_1=\sum_{n=1}^N|a_n|$, and where $N$ is the initial number of Green's functions of our dictionary.
 The new discretized objective functional is
\begin{equation}
f_1=\arg \min_{\M a,\M b}\underbrace{\left(E(\V z ,\left(\M P \M a +\M Q \M b \right))+ \lambda \|\M a\|_1 \right)}_{J_{1}(\V z|\lambda, \M a,\M b)=J_{1}(\V z|\lambda,f_1)}.
\end{equation}

When $E$ is differentiable with respect to the parameters, a minimum can be found by using proximal algorithms where the slope of $\|\M a\|_1$ is defined by a Prox operator. 
We discuss the two special cases when $E$ is either an indicator function or a quadratic data-fidelity term.

 \vspace{0 mm}
\subsubsection{Exact Fit with $E=\Spc I(\V z_0, \Op H \{f\})$}
To perfectly recover the measurements, we impose an infinite penalty when the recovered measurements differ from the given ones. In view of \eqref{eqn:exact1} and \eqref{eqn:exact2}, this corresponds to solving
\begin{align}\label{eqn:LP1}
(\M a^*, \M b^*)= \arg \min_{\M a, \M b} \|\M a\|_1 \quad \text{subject to} \quad \M P \M a +\M Q \M b=\V z.
\end{align}
We then recast Problem (\ref{eqn:LP1}) as the linear program
\begin{align}\nonumber
 (\M a^*, \M u^*, \M b^*)=\min_{\M a,\M u,\M b} ~\sum_{n=1}^N  u_n ~\text{subject~to}~ 
 \M u +\M a&\geq \M 0, \\ \nonumber
 \M u -\M a&\geq \M 0, \\
 \quad \M P \M a +\M Q \M b &=\V z, 
 \end{align}
where the inequality $\M x \geq \M y$ between any 2 vectors $\M x \in \R^N$ and $\M y \in \R^N$ means that $x_n \geq y_n$ for $n \in [1\ldots N]$. This linear program can be solved by a conventional simplex or a dual-simplex approach \cite{Dantzig1955}, \cite{Luenberger1973}. 

\subsubsection{ Least Squares Fit with $E=\|\V z -\Op H \{f\}\|^2_2$} When $E$ is a quadratic data-fidelity term, the problem becomes
\begin{align}\label{lassoNull}
 (\M a^*,\M b^*)=\arg \min_{\M a, \M b} \left(\|\V z -\left(\M P \M a +\M Q \M b \right)\|^2_2+ \lambda \|\M a\|_1\right),
\end{align}
which is more suitable when the measurements are noisy. The discrete version \eqref{lassoNull} is similar to \eqref{discrete1}, the fundamental difference being in the nature of the underlying basis function. 

The problem is converted into a LASSO formulation \cite{Tibshirani1996} by decoupling the computation of $\M a^*$ and $\M b^*$. Suppose that $\M a^*$ is fixed, then $\M b^*$ is found by differentiating (\ref{lassoNull}) and equating the gradient to~$\M 0$. This leads to
\begin{equation}
\M b^*=\left(\M {Q}^T\M Q\right)^{-1} \M {Q}^T (\V z-\M P \M a^*).
\end{equation}
Upon substitution in (\ref{lassoNull}), we get that
\begin{align}\label{lasso}
 \M a^*&=\arg \min_{\M a} \left(\|\M Q' \V z -\M Q' \M P \M a \|^2_2+ \lambda \|\M a\|_1\right),
\end{align}
where $\M Q'=\left(\M I-\M Q\left(\M {Q}^T\M Q\right)^{-1} \M {Q}^T \right)$ and $\M I$ is the $(M\times M)$ identity matrix.
Problem (\ref{lasso}) can be solved using a variety of optimization techniques such as interior-point methods or proximal-gradient methods, among others. We employ the popular iterative algorithm FISTA \cite{beck2009}, which has an $\Spc O(1/t^2)$ convergence rate  with respect to its iteration number $t$. However, in our case, the system matrices are formed by the measurements of the shifted Green's function on a fine grid. This leads to high correlations among the columns and introduces two issues.
\begin{itemize}
\item If LASSO has multiple solutions, then FISTA can converge to a solution within the solution set, whose sparsity index is greater than $M$.
\item If LASSO has a unique solution, then the convergence to the exact solution can be slow. The convergence rate is inversely proportional to the Lipschitz constant of the gradient of a quadratic loss function $\left(\max\text{Eig}\left(\M H^T\M H\right)\right)$, which is typically high for the system matrix obtained through our formulation.
\end{itemize} 

We address these issues by using a combination of FISTA and simplex, governed by the following Lemma \ref{lemma:KKT} and Theorem \ref{prop:Simplex}.
The properties of the solution of the LASSO problem have been discussed in \cite{Tibshirani2013}, \cite{Rauhut2008}, \cite{foucart2013a}. We quickly recall one of the main results from \cite{Tibshirani2013}.
\begin{lemma}[{\cite[Lemma 1 and 11]{Tibshirani2013}}]\label{lemma:KKT} Let $\V z \in \Rm$ and $\M H \in \R^{M \times N}$, where $M<N$. Then, the solution set
\begin{align}\label{eqn:Lassolemma}
\alpha_{\lambda}=\left\lbrace\arg \min_{\M a \in \R^{N}} \left(\|\V z -\M H \M a\|^2_2+ \lambda \|\M a\|_1 \right) \right\rbrace
\end{align}
has the same measurement $\M H \M a^*=\V z_0$ for any $\M a^* \in \alpha_{\lambda}$. Moreover, if the solution is not unique, then any two solutions $\M a^{(1)}, \M a^{(2)} \in \alpha_{\lambda}$ are such that their $m$th element satisfies $\left\{{\rm sign}\left(\M a_m^{(1)}\right){\rm sign}\left(\M a_m^{(2)}\right)\geq 0\right\}$ for $m \in [1\ldots M]$. In other words, any two solutions have the same sign over their common support.
\end{lemma}

We use Lemma \ref{lemma:KKT} to infer Theorem \ref{prop:Simplex}, whose proof is given in Appendix \ref{prop:Simplex}.
\begin{theorem}\label{prop:Simplex}
Let $\V z \in \Rm$ and $\M H \in \R^{M \times N}$, where $M<N$. Let $\V z_{0,\lambda}=\M H \M a^*, \forall \M a^* \in \alpha_{\lambda},$ be the measurement of the solution set $\alpha_{\lambda}$ of the LASSO formulation
 \begin{align}
\M a^*=\arg \min_{\M a \in \R^{N}} \left(\|\V z -\M H \M a\|^2_2+ \lambda \|\M a\|_1 \right).\label{Lasso}
\end{align} 
Then, the solution $\M a^*_{\rm SLP}$ (obtained using the simplex algorithm) of the linear program corresponding to the problem
\begin{align}\label{eqn:LP2}
\M a^*_{\rm SLP}=\arg \min \|\M a\|_1 \quad \text{subject~to}\quad  \M H \M a=\V z_{0,\lambda}
\end{align}
 is an extreme point of $\alpha_{\lambda} $. Moreover, $\|\M a^*_{\rm SLP}\|_0 \leq M$.
\end{theorem}
Theorem \ref{prop:Simplex} helps us to find an extreme point of the solution set $\alpha_{\lambda}$ of a given LASSO problem in the case when its solution is non-unique. 
To that end, we first use FISTA to solve the LASSO problem until it converges to a solution $\M a_{\text{F}}$. 
By setting $\V z_{0,\lambda}=\M H \M a_{\rm F}$, Lemma \ref{lemma:KKT} then implies that $\M H \M a=\V z_{0,\lambda}, \forall \M a \in \alpha_{\lambda}$. 
We then run the simplex algorithm to find
\begin{align*}
\M a_{\rm SLP}= \arg \min \|\M a\|_1 \quad \text{subject~to}\quad  \M H \M a=\M H \M a_{\text{F}},
\end{align*}
which yields an extreme point of $\alpha_{\lambda}$ by Theorem \ref{prop:Simplex}. \\
An example where the LASSO problem has a non-unique solution is shown in Figure \ref{1MS}.b. 
In this case, FISTA converges to a non-sparse  solution with $\|\M a_{\rm{F}}\|>M$, shown as solid stems. This implies that it is not an extreme point of the solution set.
The simplex algorithm is then deployed to minimize the $\ell_1$ norm such that the measurement $\V z_0=\M H \M a_{\rm{F}}$ is preserved.
The final solution shown as dashed stems is an extreme point with the desirable level of sparsity. The continuous-domain relation of this example is discussed later. \\
The solution of the continuous-domain formulation is a convex set whose extreme points are composed of at most $M$ shifted Green's functions. To find the position of these Green's functions, we discretize the continuum into a fine grid and then run the proposed two-step algorithm. If the discretization is fine enough, then the continuous-domain function that corresponds to the extreme point of the LASSO formulation is a good proxy for the actual extreme point of the convex-set solution of the original continuous-domain problem. This makes the extreme-point solutions of the LASSO a natural choice among the solution set.
\begin{figure}[t]
\begin{minipage}[b]{\linewidth}
  \centering
  \centerline{ \includegraphics[width=\linewidth,height=4.5cm ]{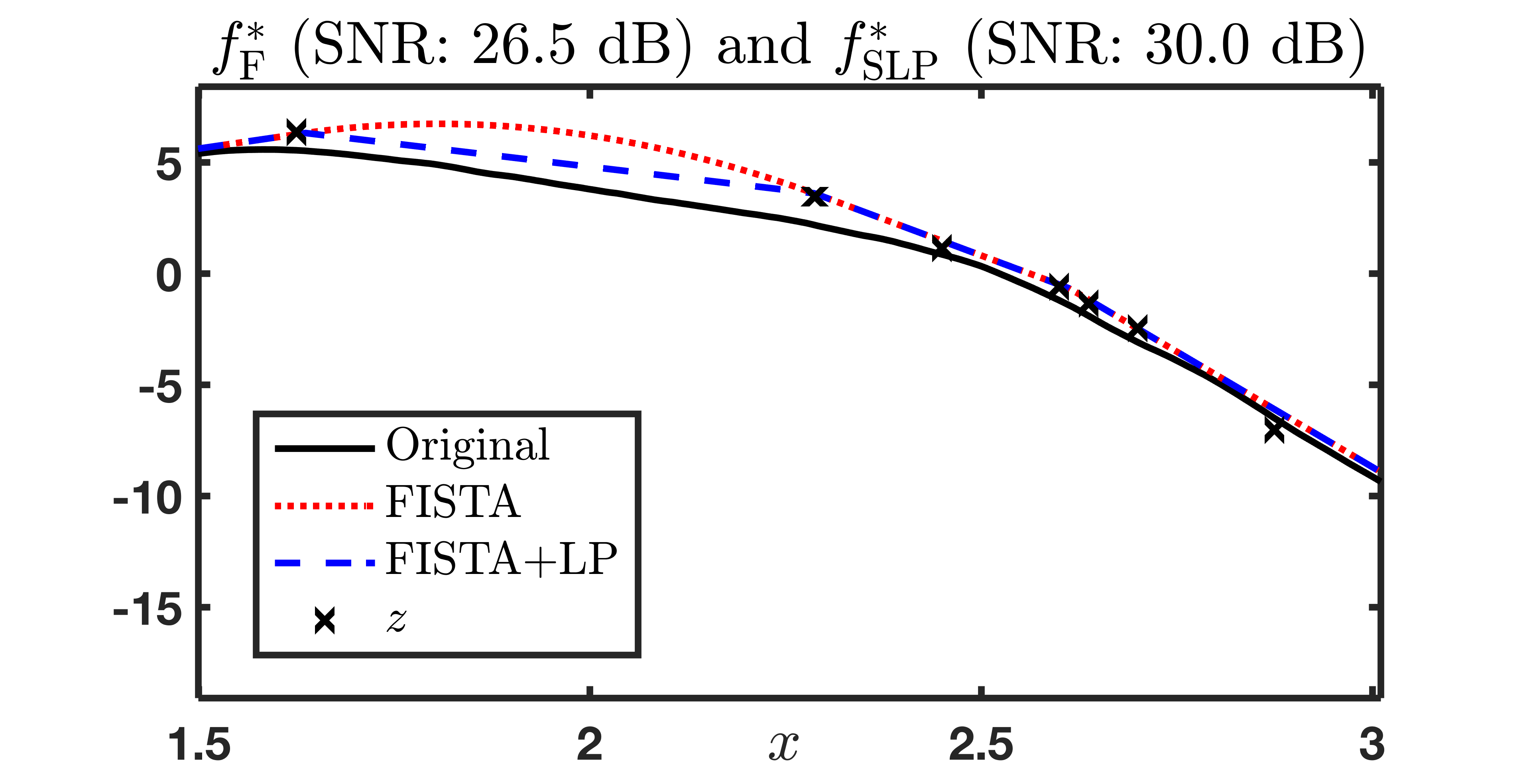}}
 \centerline{(a) }
\end{minipage}
\begin{minipage}[b]{\linewidth}
\vspace{3mm}
  \centerline{ $\M a^*_{\rm F}$ and $\M a^*_{\rm SLP}$}
  \centerline{\includegraphics[width=\linewidth, height=4.1cm, trim={0 0 0 5mm },clip]{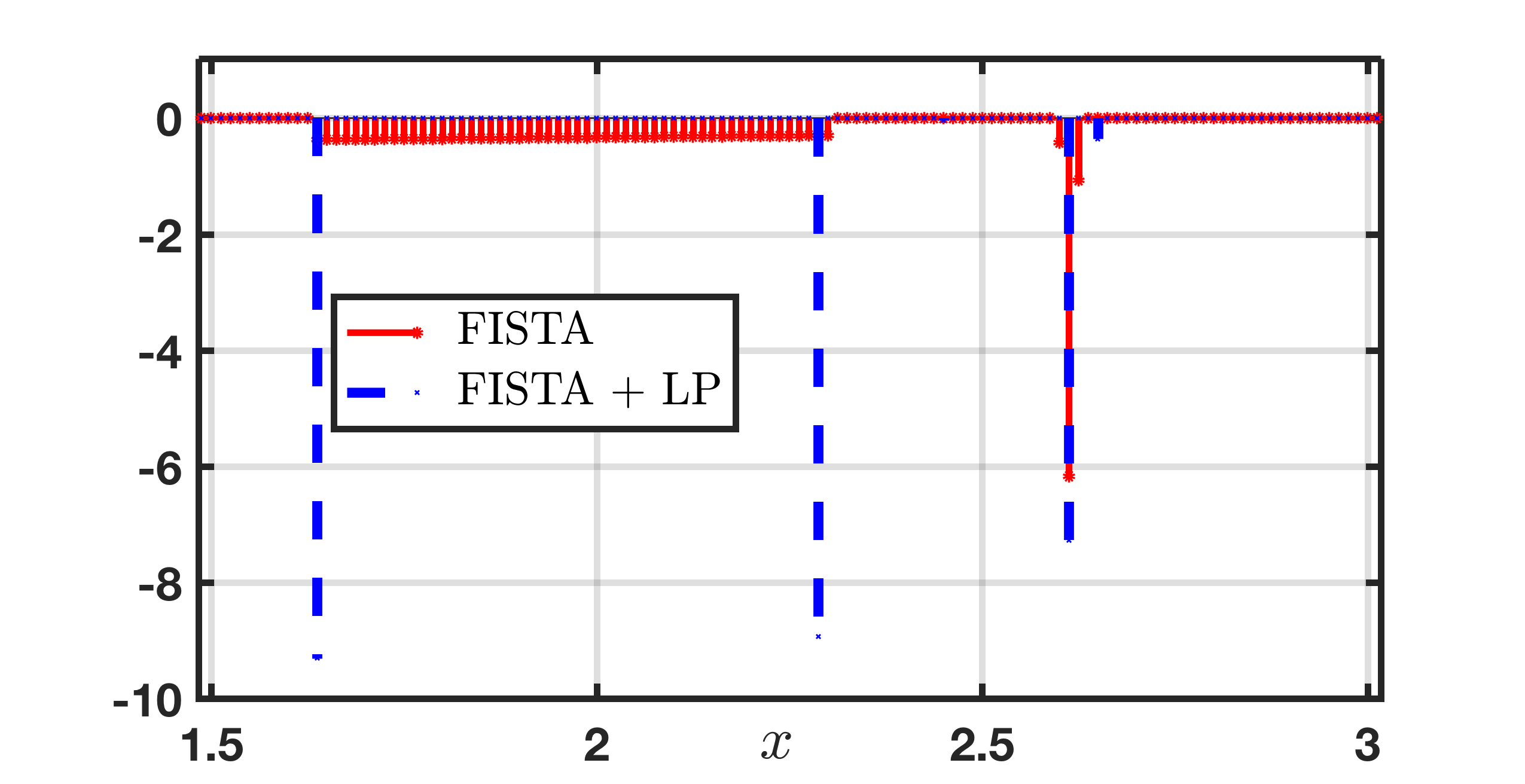}}
\centerline{(b)}
\end{minipage}
\begin{minipage}[b]{\linewidth}
\vspace{3mm}
  \centering
  \centerline{\includegraphics[width=\linewidth, height=4.5cm ]{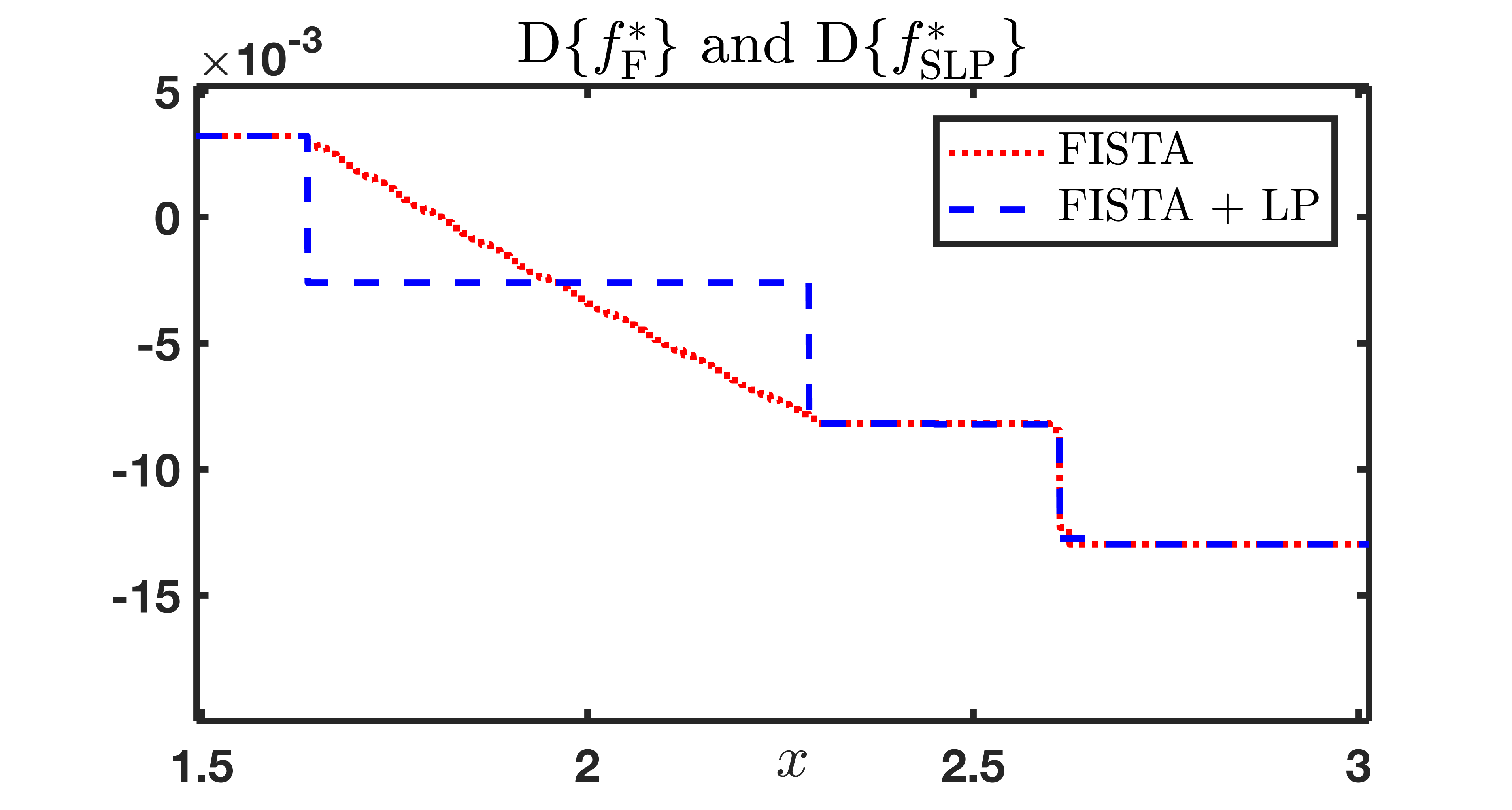}}
  \centerline{(c) }
\end{minipage}
\caption{Illustration of inability of FISTA to deliver a sparse solution : (a) comparison of solutions, $f^*_{\rm F}$ vs. $f^*_{\rm SLP}$ for continuous-domain gTV problem, (b) signal innovations with sparsity index 64 ($>M$) and 21 ($<M$), respectively, and (c) derivative of the two solutions. The two signal innovations in (b) are solutions of the same Lasso problem, but only $\M a_{\rm SLP}$ is an extreme point of the solution set. The original signal is a second-order process ($\Lop =\Op D^2$) and the measurements are $M=30$ nonuniform noisy samples (SNR = 40 dB). The parameters are $\lambda=0.182$, $N=400$, and grid step $\Delta=\frac{1}{80}$.  \label{1MS}\label{2MS} }
\end{figure}
%
For the case when there is a unique solution but the convergence is too  slow owing to the high value of the Lipschitz constant of the gradient of the quadratic loss, the simplex algorithm is used after the FISTA iterations are stopped using an appropriate convergence criterion. For FISTA, the convergence behavior is ruled by the number of iterations $t$ as
\begin{equation}
F(\M a_t)-F(\M a^*) \leq \frac{C}{(t+1)^2},
\end{equation}
where $F$ is the LASSO functional and 
\begin{align}
 C=2\|\M a_0- \M a^*\|_2^2 \max{\text{Eig}\left(\M H^T\M H\right)}
\end{align} (see \cite{beck2009}). This implies that an $\epsilon$ neighborhood of the minima of the functional is obtained in at most $t=\sqrt{C/\epsilon}$ iterations. However, there is no direct relation between the functional value and the sparsity index of the iterative solution. Using the simplex algorithm as the next step guarantees the upper bound $M$ on the sparsity index of the solution. Also, $F(\M a_{\text{SLP}})\leq F(\M a_{\text{F}})$. This implies that an $\epsilon$-based convergence criterion, in addition to the sparsity-index-based criterion like $\M a_{\text{F}} \leq M$, can be used to stop FISTA. Then, the simplex scheme is deployed to find an extreme point of the solution set with a reduced sparsity index.

\section{Illustrations}
We discuss the results obtained for the cases when the measurements are random samples either of the signal itself or of its continuous-domain Fourier transform. The operators of interest are $\Lop=\Op D$ and $\Lop=\Op D^2$. The test signal $f$ is solution of the stochastic differential equation $\Lop f=w$ \cite{unser2005generalized} for the two cases when $w$ is
\begin{itemize}
\item  \textbf{Impulsive Noise}. Here, the innovation $w$ is a sum of Dirac impulses whose locations follow a compound-Poisson distribution and whose amplitudes follow a Gaussian distribution. The corresponding process $s$ has then the particularity of being piecewise smooth \cite{Unser2011stochastic}.  
This case is matched to the regularization operator $\|\Lop f \|_{\Spc M}$ and is covered by Theorem \ref{theo:L1_representer} which states that the minima $f^*_{\text{gTV}}$ for this regularization case is such that 
\begin{align} 
w^*_{\text{gTV}}=\Lop f^*_{\text{gTV}} =\sum_{k=1}^K a_{k} \delta (\cdot - x_k),
\end{align}
which is a form compatible with a realization of an impulsive white noise.
\item \textbf{Gaussian White Noise}. This case is matched to the regularization operator $\|\Lop f \|_{L_2}$. Unlike the impulsive noise, $w^*_{L_2}=\Lop f^*_{L_2}$ is not localized to finite points and therefore is a better model for the realization of a Gaussian white noise. 
\end{itemize}
In all experiments, we also constrain the test signals to be compactly supported. This can be achieved by putting linear constraints on the innovations of the signal. 
In Sections \ref{sec:rs} and \ref{sec:fs}, we confirm experimentally that matched regularization recovers the test signals better than non-matched regularization. 
While reconstructing the Tikhonov and gTV solutions when the measurements are noisy, the parameter $\lambda$ in  \eqref{eqn:L2matrix} and \eqref{lassoNull} is tuned using a grid search to give the best recovered SNR.

\subsection{Random Sampling}\label{sec:rs}
In this experiment, the measurement functionals are Dirac impulses with the random locations $\{x_m\}_{m=1}^M$. The regularization operator is $\Lop=\Op D^2$. It corresponds to $\rho_{\Op D^2}(x)=-\frac{1}{2}|x|$ and $\varphi_{\Op D^2}(x)=\left(\rho_{\Lop^*\Lop}*h_m\right)(x)=|x-x_m|^3/12$. The null space is $\Spc N_{\Op D^2}={\rm span}\{1,x\}$ for this operator.
This means that the gTV-regularized solution is piecewise linear and that the $L_2$-regularized solution is piecewise cubic. We compare in Figures \ref{2SS}.a and \ref{2SS}.b the recovery from  noiseless samples of a second-order process, referred to as ground truth (GT). It is composed of sparse (impulsive Poisson) and non-sparse (Gaussian) innovations, respectively \cite{Unser2014book}.  The sparsity index---the number of impulses or non-zero elements---for the original sparse signal is 9. 
The solution for the gTV case is recovered with $\Delta=$ 0.05 and $N= 200$. The sparsity index of the gTV solution for the sparse and Gaussian cases are 9 and 16, respectively. As expected, the recovery of the gTV-regularized reconstruction is better than that of the $L_2$-regularized solution when the signal is sparse. For the Gaussian case, the situation is reversed.
\begin{figure*}
\begin{minipage}[b]{0.48\linewidth}
  \centering
  \centerline{ \includegraphics[width=\linewidth,height=5cm ]{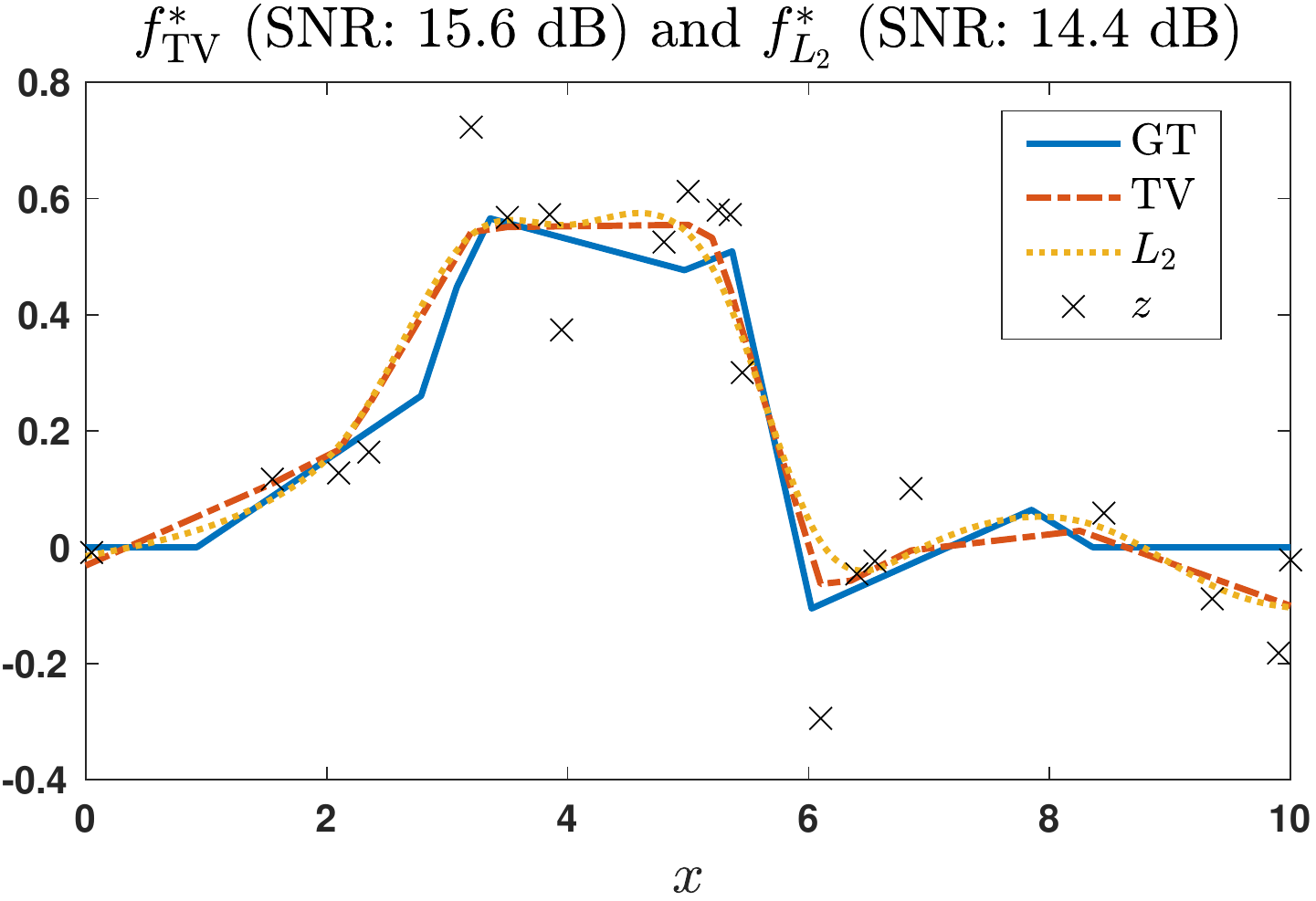}}
  \centerline{(a) Sparse Signal}
\end{minipage}
\begin{minipage}[b]{0.48\linewidth}
  \centering
  \centerline{\includegraphics[width=\linewidth,height=5cm ]{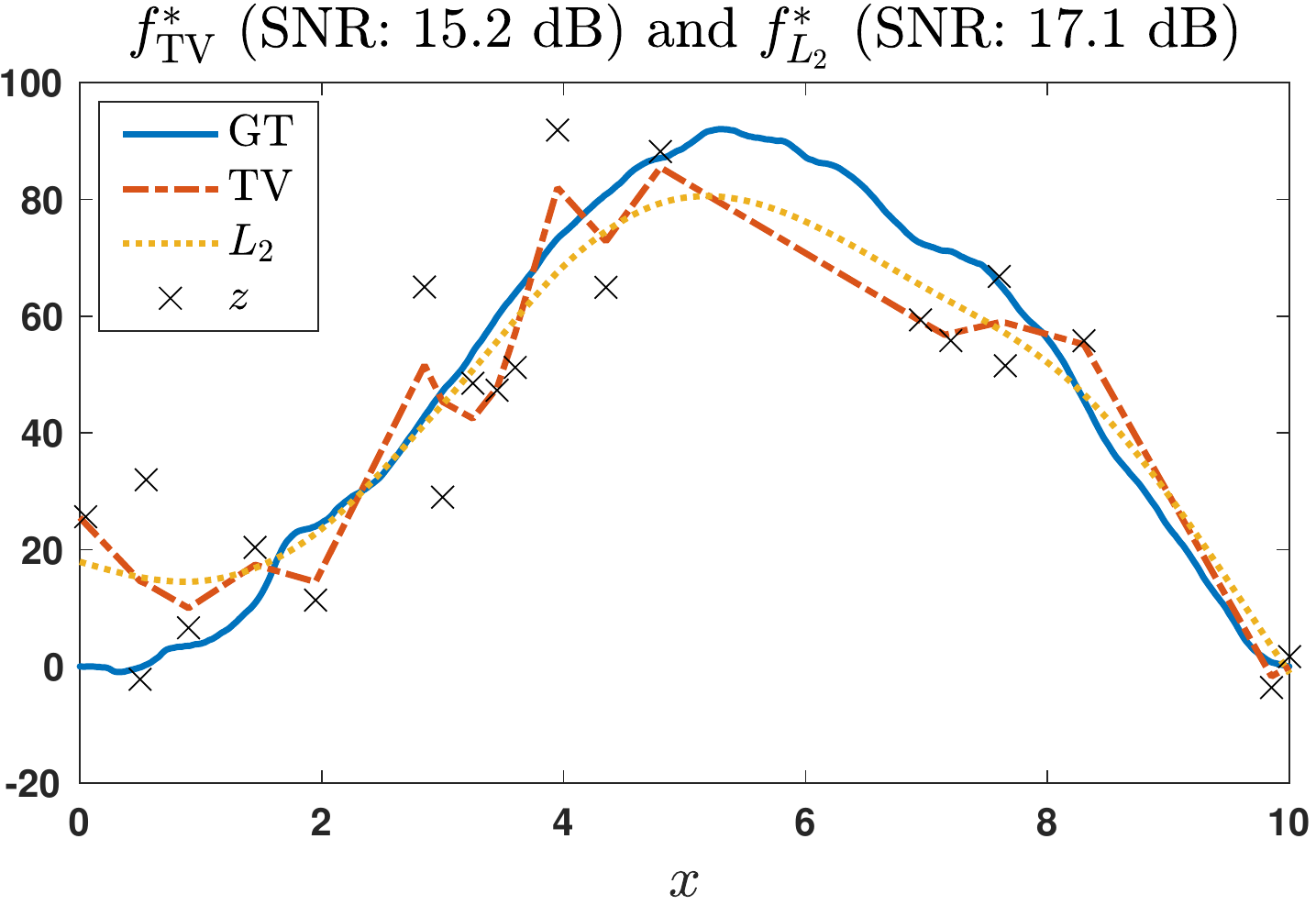}}
  \centerline{(b) Gaussian Signal}
\end{minipage}
\caption{Recovery of sparse (a) and Gaussian (b) second-order processes (GT) using $\Lop=\Op D^2$ from their nonuniform samples corrupted with 40 dB measurement noise.}\label{2SS}	
\end{figure*}

\subsection{Multiple Solutions}
We discuss the case when the gTV solution is non-unique. We show in Figure \ref{2MS}.a examples of solutions of the gTV-regularized random-sampling problem obtained using FISTA alone $(f_{\text{F}})$ and FISTA + simplex (linear programming, $f_{\text {SLP}}$). In this case, $M=30$, $\Lop=\Op D^2$, and $\lambda=0.182$. The continuous-domain functions $f_{\text F}$ and $f_{\text {SLP}}$ have basis functions whose coefficients are the (non-unique) solutions of a given LASSO problem, as shown in Figure \ref{1MS}.b. 
The $\ell_1$ norms of the corresponding coefficients are the same. Also, it holds that
\begin{align}
\|\Op D^2f_{\text F}\|_{\Spc M}&=\|\Op D^2 f_{\text {SLP}}\|_{\Spc M}
=\|\Op D f_{\text F}\|_{\text{TV}}=\|\Op D f_{\text {SLP}}\|_{\text{TV},}
\end{align} 
which implies that the TV norm of the slope of $f_{\rm F}$ and $f_{\rm {SLP}}$ are the same. This is evident from Figure \ref{2MS}.c.
The arc-length of the two curves are the same. 
The signal $f_{\text {SLP}}$ is piecewise linear ($21< M$), carries a piecewise-constant slope, and is by definition, a non-uniform spline of degree 1. By contrast, $f_{\rm F}$ has many more knots and even sections whose slope appears to be piecewise-linear.

Theorem \ref{theo:L1_representer} asserts that the extreme points of the solution set of the gTV regularization need to have fewer than $M$ knots. 
Remember that $f_{\text {SLP}}$ is obtained by combining FISTA and simplex; this ensures that the basis coefficients of $f_{\text {SLP}}$ are the extreme points of the solution set of the corresponding LASSO problem (Theorem \ref{prop:Simplex}) and guarantees that the number of knots is smaller than $M$. 

This example shows an intuitive relationship between the continuous-domain and the discrete-domain formulations of inverse problems with gTV and $\ell_1$ regularization, respectively. The nature of the continuous-domain solution set and its extreme points resonates with its corresponding discretized version. In both cases, the solution set is convex and the extreme points are sparse.

\begin{figure*}	[h]  
\begin{minipage}[b]{0.48\linewidth}
 		\centerline{ \includegraphics[width=0.95\linewidth,height=5cm]{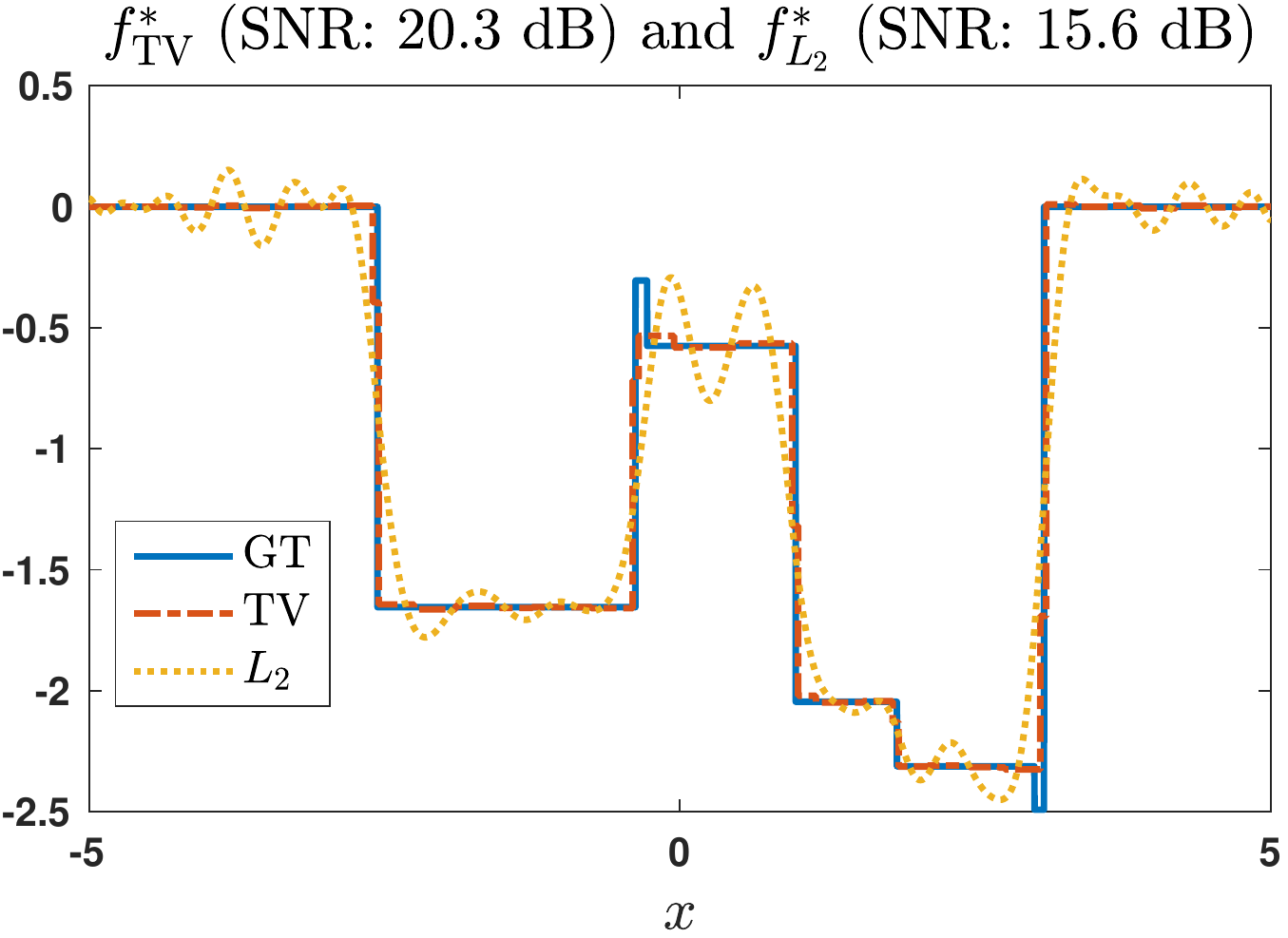}}
 		 \centerline{(a) Sparse Signal} 
\end{minipage}
\begin{minipage}[b]{0.48\linewidth}
 		\centerline{  \includegraphics[width=0.95\linewidth,height=5cm]{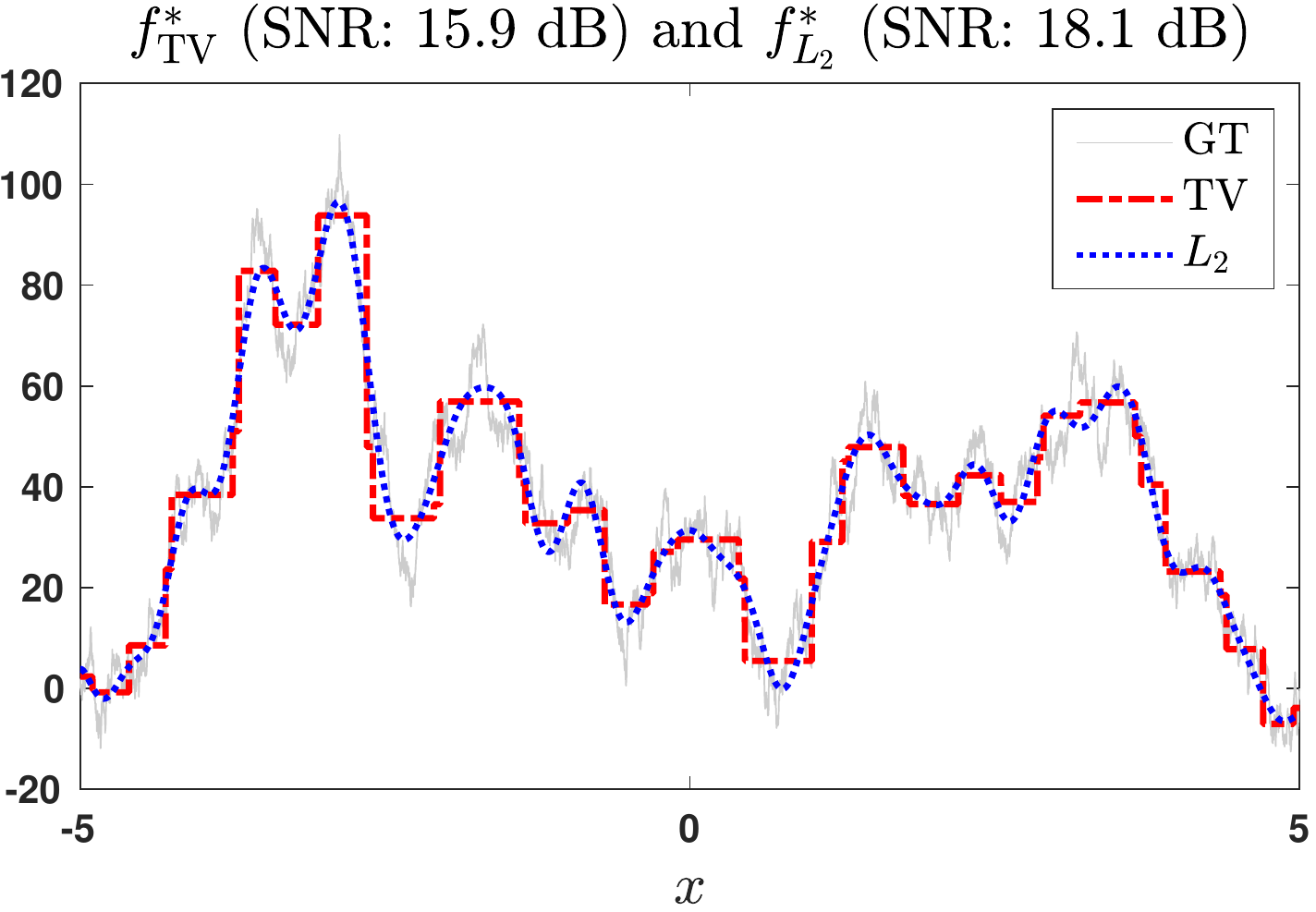}}
 		    \centerline{(b) Gaussian Signal} 
\end{minipage}
     
\begin{minipage}[b]{0.48\linewidth}
\vspace{5mm}
 		\centerline{\includegraphics[width=0.95\linewidth,height=5cm]{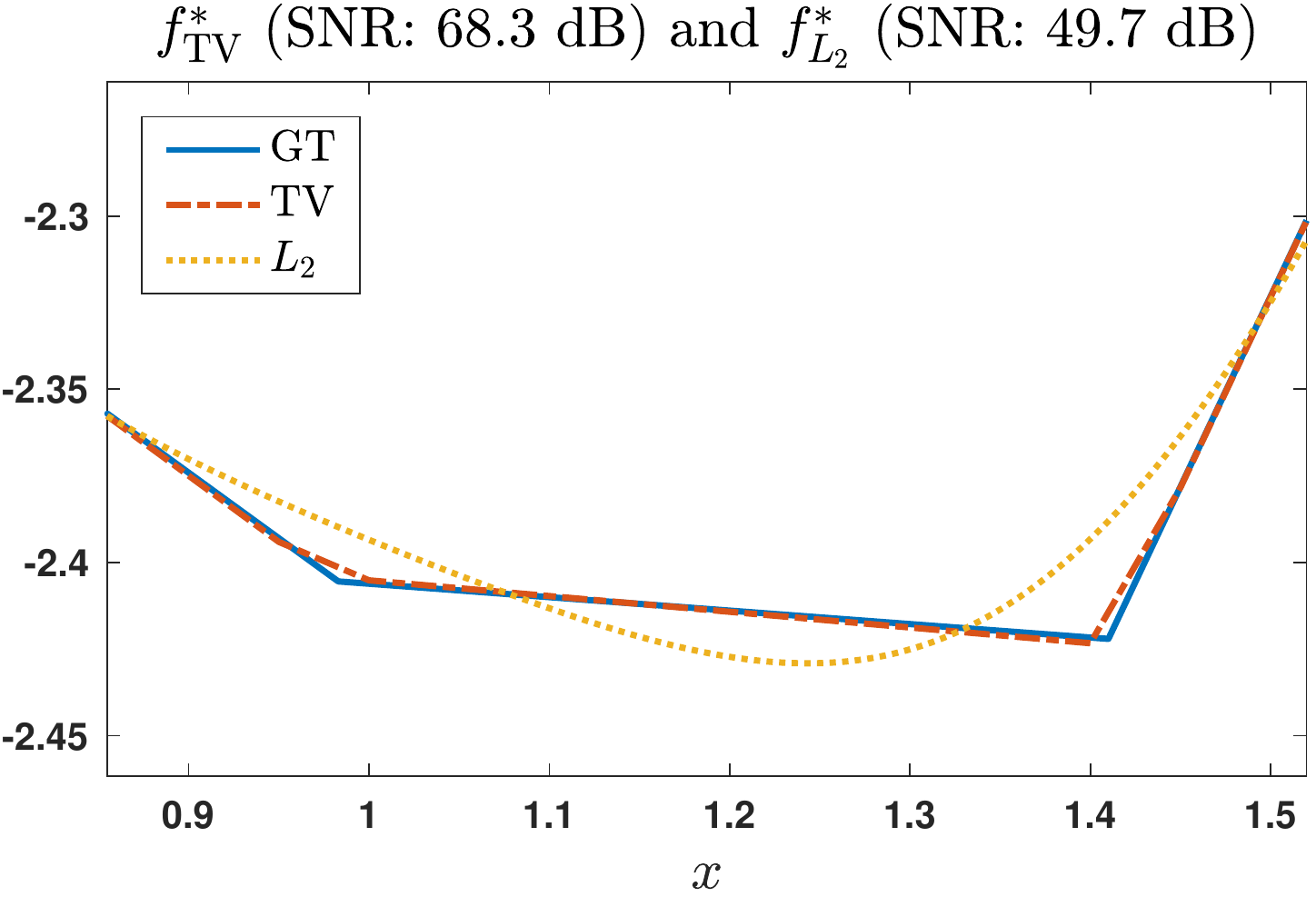}}
 		    \centerline{(c) Sparse Signal}    
\end{minipage}
\begin{minipage}[b]{0.48\linewidth}

 		\centerline{\includegraphics[width=0.95\linewidth,height=5cm]{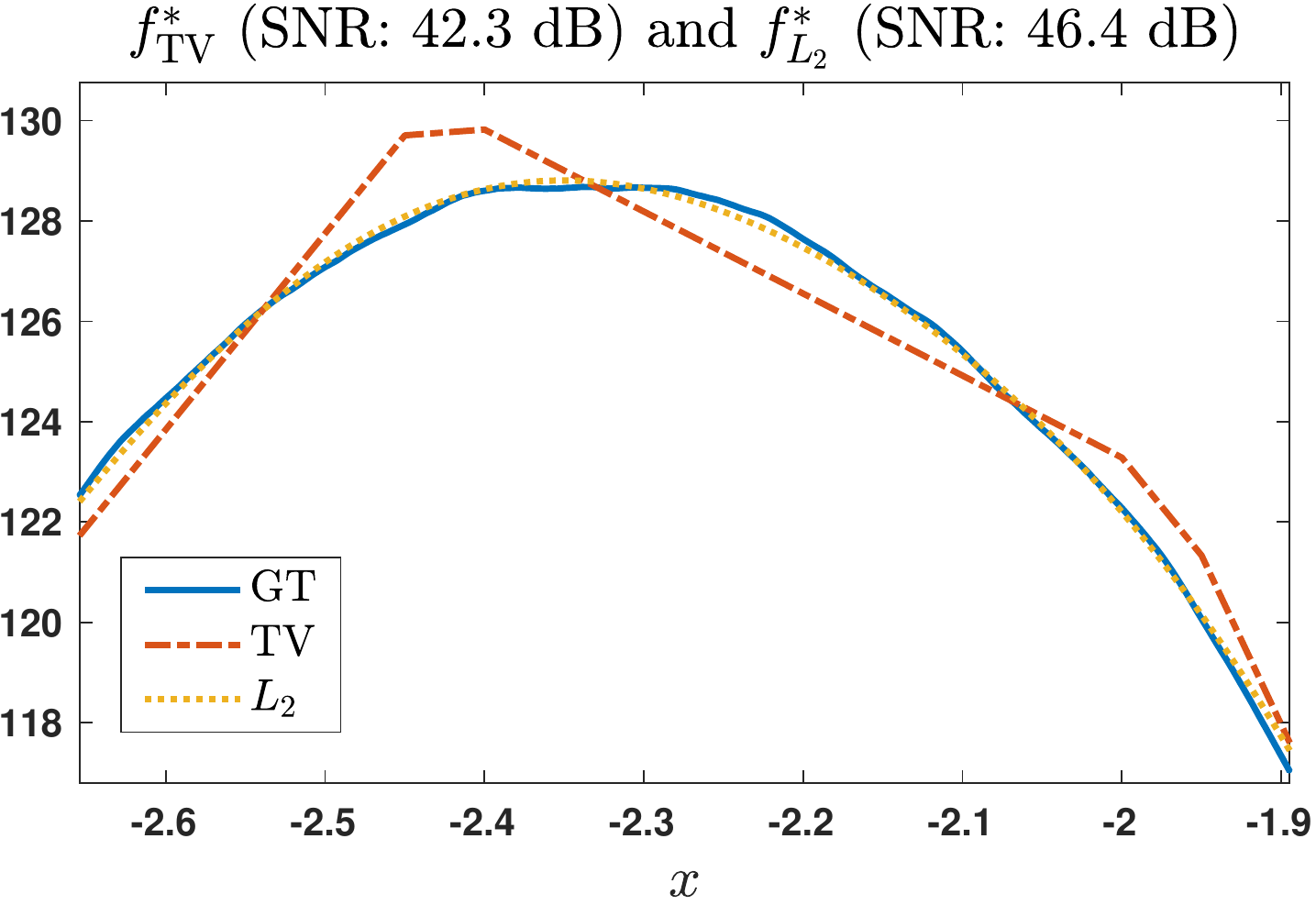}}
 		    \centerline{(d) Gaussian Signal}     
\end{minipage}
\caption{Recovery of first-order (first row) and second-order (second row) processes from their random noiseless Fourier samples. In all the cases, $M=41$ and $N=200$. In the interest of clarity, (c) and (d) contain the zoomed versions of the actual signals. \label{1FS}} 
\end{figure*}


\subsection{Random Fourier Sampling}\label{sec:fs}
Let now the measurement functions be $h_{m}(x)={\rm rect}\left(\frac{x}{T}\right)\mathrm{e}^{-{\mathrm {j}} \omega_{m}x}$, where $T$ is the window size. The samples are thus random samples of the continuous-domain Fourier transform of a signal restricted to a window. For the regularization operator $\Lop=\Op D$, the Green's function is $\rho_{\Op D}(x)=\mathds{1}_+(x)$ and the basis is $\varphi_{\Op D,m}(x)=\left(\frac{1}{2}|\cdot|*h_m\right)(x)$. 

Figure \ref{1FS}.a and \ref{1FS}.b correspond to a first-order process with sparse and Gaussian innovations, respectively. The grid step $\Delta= 0.05$, $M=41$, and $N=200$. The sparsity index of the gTV solution for the sparse and Gaussian cases is 36 and 39, respectively. For the original sparse signal (GT), it is 7. The oscillations of the solution in the $L_2$-regularized case are induced by the sinusoidal form of the the measurement functionals. This also makes the $L_2$ solution intrinsically smoother than its gTV counterpart. Also, the quality of the recovery depends on the frequency band used to sample.

In Figures \ref{1FS}.c and \ref{1FS}.d, we show the zoomed version of the recovered second-order process with sparse and Gaussian innovations, respectively. The grid step is $\Delta= 0.05$, $M=41$ and $N=200$. 
The operator $\Lop=\Op D^2$ is used for the regularization. 
This corresponds to $\rho_{\Op D^2}(x)=x_+$ and $\varphi_{\Op D^2,m}(x)=\left(\frac{1}{12}|\cdot|^3*h_m\right)(x)$. 
The sparsity index of the gTV solution in the sparse and Gaussian cases is 10 and 36, respectively. For the original sparse signal (GT), it is 10. Once again, the recovery by gTV is better than by $L_2$ when the signal is sparse. In the Gaussian case, the $L_2$ solution is better.
 
The effect of sparsity on the recovery of signals from their noiseless and noisy (40 dB SNR) Fourier samples are shown in Table 1. The sample frequencies are kept the same for all the cases. Here, $M=41$, $N=200$, $T=10$, and the grid step  $\Delta=0.05$. We observe that reconstruction performances for random processes based on impulsive noise are  comparable to that of Gaussian processes when the number of impulses increases. This is reminiscent of the fact that generalized-Poisson processes with Gaussian jumps are converging in law to corresponding Gaussian processes \cite{Fageot2016gaussian}.

 \begin{figure*}[htbp]
 \centerline{ \includegraphics[width=0.6\linewidth,height=8.75 cm,trim={0 7mm 2cm 7mm},clip]{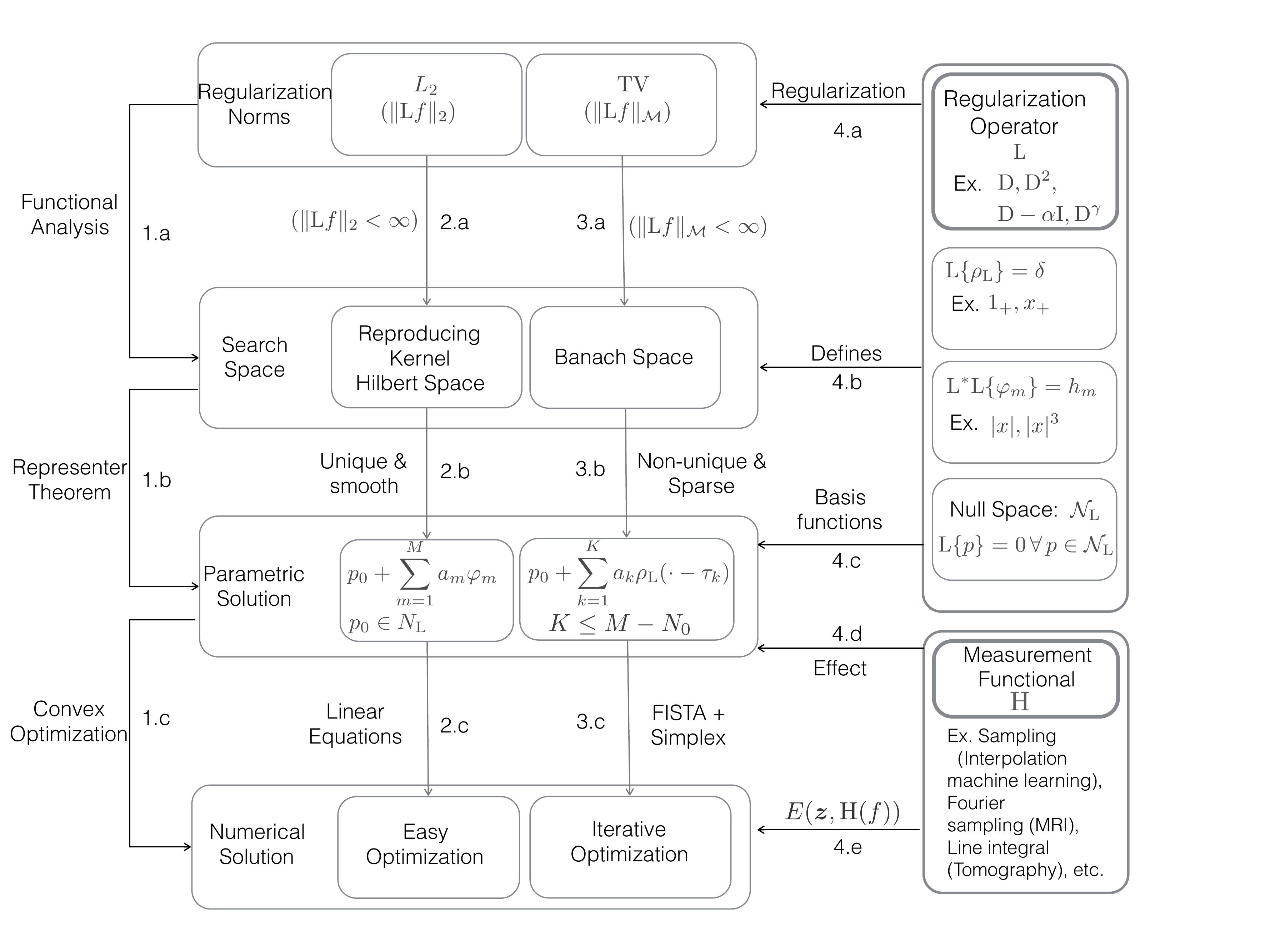}}
\caption{  
 Summary of the whole scheme. The regularization operator with a given norm \{4.a\} defines the search space for the solution\{1.a, 4.b\}. Representer theorems then give the parametric representation of the solution \{1.b\}. The numerical solution is then recovered by optimizing over the parameters to minimize $J_R(\V z|f)$ \{1.c\}. \label{Flowchart}}
\end{figure*}

\section{Conclusion}
We have shown that the formulation of continuous-domain linear inverse problems with Tikhonov- and total-variation-based regularizations leads to spline solutions. The nature of these splines is dictated by the Green's function of the regularization operator $\Lop$ and $(\Lop^*\Lop) $ for Tikhonov and total variation, respectively.
The former is better to reconstruct smooth signals; the latter is an attractive choice to reconstruct signals with sparse innovations. Representer theorems for the two cases come handy in the numerical reconstruction of the solution. They allow us to reformulate the infinite-dimensional optimization as a finite-dimensional parameter search. The formulations and the results of this paper are summarized in Figure \ref{Flowchart}.

\begin{table*}[]

\hspace{5mm}
\begin{tabular}{l|l|ll|ll}
\hline
\hline
No. of &  & \multicolumn{2}{l|}{~~~$\Op D$} & \multicolumn{2}{l}{~~~$\Op D^2$} \\ \cline{3-6} 
impulses  &Sparsity & TV   & $L_2$ &TV   & $L_2$  \\ \hline
\multicolumn{1}{r|}{10} & Strong & \textbf{19.60} & 15.7 & \textbf{52.08} & 41.54  \\
\multicolumn{1}{r|}{100} & Medium & \textbf{16.58}  &{16.10}  & \textbf{41.91}  & 41.26  \\
\multicolumn{1}{r|}{2000} & Low &  14.45&\textbf{16.14}  & 39.68  &\textbf{41.40}  \\
\multicolumn{1}{r|}{-} & Gaussian & 14.30 &\textbf{16.32}  &40.05  & \textbf{41.23} \\ \hline
\end{tabular}
\hspace{5mm}
\begin{tabular}{l|l|ll|ll}
\hline
\hline
No. of &  & \multicolumn{2}{l|}{~~~$\Op D$} & \multicolumn{2}{l}{~~~$\Op D^2$} \\ \cline{3-6} 
impulses &Sparsity & TV   & $L_2$ &TV   & $L_2$  \\ \hline
\multicolumn{1}{r|}{10} & Strong & \textbf{17.06}  & 11.52 & \textbf{25.55}&24.60  \\
\multicolumn{1}{r|}{100} & Medium & \textbf{13.24} &  10.94& \textbf{24.44} &24.24  \\
\multicolumn{1}{r|}{2000} & Low & 10.61 &\textbf{11.13}  &25.80  & \textbf{26.19}  \\
\multicolumn{1}{r|}{-} & Gaussian & 10.40 &\textbf{11.10}  &24.95  & \textbf{25.48}  \\ \hline
\end{tabular}
\caption{Comparison of TV and $L_2$  recovery from their (left table) noiseless  and (right table) noisy (with 40 dB SNR) random Fourier samples. The results have been averaged over 40 realizations. \label{Noiseless}}

\end{table*}

\begin{appendices}

\section{Proof of Theorem \ref{solution set} }\label{appx:solutionset}
Let $J^*$ be the minimum value attained by the solutions. Let $f_1$ and $f_2$ be two solutions. Let $E_1$, $E_2$ be their corresponding $E$ functional value and  let $R_{1}, R_2$ be their corresponding regularization functional value. Since the cost function is convex, any convex combination $f_{12}=\beta f_1 +(1-\beta)f_2$ is also a solution for $\beta \in [0,1]$ with functional value $J^*$. Let us assume that $\Op H \{f_1\}\neq \Op H\{ f_2\}$. Since $E$ is strongly convex and $R$ is convex, we get that
\begin{IEEEeqnarray}{rcl}\nonumber
 J(f) &=&E(\V z, \Op H\{\beta f_1+(1-\beta)f_2\})+\lambda R( \beta f_{1} + (1-\beta) f_2) \nonumber\\ 
&<& \underbrace{\beta E_1 +(1-\beta) E_2 +\beta R_{1}+(1-\beta) R_2}_{ J^*}. \nonumber\\ \nonumber 
\end{IEEEeqnarray}
This is a contradiction. Therefore, $\Op H\{f_1\}= \Op H\{f_2\}=\Op H\{f_{12}\}$.

\section{Abstract Representer Theorem}

The result presented in this section is preparatory to Theorem \ref{theo:L2_representer}. It is classical for Hilbert spaces. We give its proof for the sake of completeness. 

\begin{theorem}\label{theo:ART}
 Let $\Spc X$ be a Hilbert space equipped with the inner product $\langle \cdot ,\cdot  \rangle_{\Spc X}$ and a set of linear functionals $h_1,\ldots,h_M \in \Spc X'$. Let  $\Spc C  \in \R^M$ be a feasible convex compact set, meaning that there exists at least a function $f \in \mathcal{X}$ such that $\Op H \{f\} \in \Spc C$. Then, the minimizer 
\begin{equation}
f^*=\arg \min_{f\in \Spc X} \|f\|^2_{\Spc X} \text{~s.t.~} \Op H\{ f\}\in \Spc C
\end{equation}
exists, is unique, and can be written as
\begin{equation}
f^*=\sum_{m=1}^M a_m h^*_m
\end{equation}
for some $\{a_m\}^M_{m=1} \in \R$, where $h^*_m=\Op R h_m$ and $\Op R: \Spc X' \to \Spc X$ is the Riesz map of $\Spc X$.
\end{theorem}

\begin{proof}
Let $\Spc C_{\Spc X} = \Op H^{-1} ( \Spc C) = \{ f \in \Spc X , \ \Op H\{ f\} \in \Spc C\}  \in \Spc X$, assumed to be nonempty. Since $\Op H$ is linear and bounded and  since $\Spc C$ is convex and compact, its preimage $\Spc C_{\Spc X}$ is also convex and closed.  By Hilbert's projection theorem \cite{rudin1987real}, the solution $f^*$ exists and is unique as the projection of the null function onto $\Spc C_{\Spc X}$. Let the measurement of this unique point $f^*$ be $\Op H \{f^*\}=\V z_0$.\\
The Riesz representation theorem states that $\langle h_m,f \rangle=\langle h^*_m,f \rangle_{\Spc X}$ for every $f \in \Spc X$, where $h_m^* \in \Spc X$ is the unique Riesz conjugate of the functional $h_m$. 
We then uniquely decompose $f^*$ as $f^*=f^{\bot}+\sum_{m=1}^{M} a_m h_m^* $ , where  $ f^{\bot}$ is orthogonal to the span of the $h^*_m$ with respect to the inner product on $\Spc X$.
The orthogonality implies that
\begin{equation}
 \|f^*\|^2_{\Spc X}=\left\| f^{\bot}\right\|^2_{\Spc X}+\left\|\sum_{m=1}^{M} a_m h^*_m \right\|^2_{\Spc X}.
\end{equation} 
This means that the minimum norm is reached when $f^{\bot} = 0$, implying that the form of the solution is $f^*=\sum_{m=1}^{M} a_m h_m^*$.  \\ 
\end{proof}

\section{Proof of Theorem \ref{theo:L2_representer}}    \label{apppx:L2}                                                                            
The proof of Theorem \ref{theo:L2_representer} has two steps. We first show that there exists a unique solution. Then, we use Theorem \ref{theo:ART} to deduce the form of the solution. 

\textit{Existence and Unicity of the Solution.}
As is classical in convex optimization, it suffices to show that the functional $J_2(\V z| \cdot)$ is coercive and strictly convex. 
We start with the coercivity. The measurement operator $\Op H$ is continuous and linear from $\Spc X_2$ to $\R^M$; hence, there exists a constant $C$ such that 
\begin{equation}\label{eq:1forcoercive}
\lVert \Op H \{f\} \rVert_2 \leq C \lVert f \rVert_{\mathcal{X}_2}
\end{equation}
for every $f \in \Spc X_2$. Likewise,
the condition $\Op H \{p\} = \Op H \{q\} \Rightarrow p = q$ for $p,q \in \mathcal{N}_{\Op L}$ implies the existence of $B>0$ such that \cite[Proposition 8]{Unser2016}
\begin{equation}\label{eq:2forcoercive}
\lVert \Op H \{p\} \rVert_2 \geq B \lVert p \rVert_{\mathcal{N}_{\Op L}}
\end{equation}
for every $p \in \Spc N_{\Op L}$.
Any $f \in \Spc X_2$ can be uniquely decomposed as $f = \Op L^{-1} w + p$ with $w \in  L_2 (\R)$ and $p \in \Spc N_{\Op L}$. Then, we remark that 
$\lVert f - p \rVert_{\Spc X_2} = \lVert w \rVert_{L_2}$.

Putting \eqref{eq:1forcoercive} and \eqref{eq:2forcoercive} together, we deduce with the triangular inequality that
\begin{align}
	\lVert \Op H \{f\} \lVert_2 &= \lVert \Op H \{p\} + \Op H \{ f-p \} \rVert_2  \\
	&\geq \lVert \Op H \{p\} \rVert_2 - \lVert \Op H \{ f-p\} \rVert_2   \nonumber \\
			& \geq B  \lVert p \rVert_{\Spc N_{\Op L}} - C \lVert f - p \lVert_{\Spc X_2} =  B \lVert p \rVert_{\Spc N_{\Op L}} - C \lVert w \lVert_{L_2}.  \label{eq:3forcoercive}
\end{align}
Assume that $\lVert f \rVert_{\mathcal{X}_2} \rightarrow \infty$. It means that $ \lVert p \rVert_{\Spc N_{\Op L}} $ or $ \lVert w \lVert_{L_2}$ are unbounded. 
If $ \lVert p \rVert_{\Spc N_{\Op L}} $ is significantly larger than $ \lVert w \lVert_{L_2}$, then $\lVert \Op H \{f\} \rVert \rightarrow \infty$ according to \eqref{eq:3forcoercive}; hence, $J_2(\V z|f) \geq E(\V z, \Op H \{f\}) \rightarrow \infty$ using the coercivity of $E$. 
Otherwise, it means that $ \lVert w \lVert_{L_2}$ is dominating and $J_2(\V z|f) \geq \lambda \lVert w \lVert_{L_2} \rightarrow \infty$. In both cases, $J_2(\V z|f) \rightarrow \infty$ and $J_2(\V z|\cdot)$ is coercive.\\
For the strict convexity, we first remark that $J_2(\V z| \cdot)$ is convex. For $\beta \in (0,1)$, $f_1, f_2 \in \Spc X_2$, we denote $f_{12} = \beta f_1 + (1-\beta) f_2$. 
Then, the equality case $J_2(\V z| f_{12}) = \beta J_2(\V z| f_{1 }) + (1-\beta) J_2(\V z| f_{2})$ implies that
$E(\V z| f_{12}) = \beta E(\V z| f_{1 }) + (1-\beta) E(\V z| f_{2})$ and $\lVert \Op L f_{12}  \rVert_{L_2} = \beta \lVert \Op L  f_{1} \rVert_{L_2} + (1-\beta) \lVert \Op L  f_{2} \rVert_{L_2}$, since the two parts of the functional are themselves convex. The strict convexity of $E( \boldsymbol{z} |  \cdot )$ and the norm $\lVert \cdot \rVert_2$ then implies that 
\begin{align}
 \Op L f_1  &=  \Op L f_2  \text{ and } \Op H \{f_1\}  =  \Op H \{f_2\}
\end{align}
and, therefore, $(f_1 - f_2) \in \mathcal{N}_{\Op L} \cap \mathcal{N}_{\Op H}$. Hence, $f_1 = f_2$ and the strict convexity is demonstrated. 
The functional  $J_2(\V z|\cdot)$ is coercive and strictly convex and, therefore, admits a unique minimizer $f^* \in \Spc X$.  

\textit{Form of the Minimizer.} 
Let $\V z_0 = \Op H \{f^*\}$. 
One decomposes again $\Spc X_2$ as the direct sum $\Spc X_2 = \Spc H + \Spc N_{\Op L}$, where $$\Spc H = \{ f \in \Spc X_2, \ \langle f, p \rangle = 0, \, \forall p \in \Spc N_{\Op L} \}$$ is the Hilbert space with norm $\lVert \Op L \cdot \rVert_{L_2}$.
In particular, we have that $f^* = h^* + p^*$ with $h^* \in \Spc H$ and $p^* \in \Spc N_{\Op L}$. 
Consider the optimization problem
\begin{equation} \label{eq:intermediateminimization}
\arg \min_{g\in \Spc H} \| \Op L g \|^2_{L_2} \text{~s.t.~} \Op H \{g\} = (\V z_0 - \Op H \{p^*\}),
\end{equation}
which is well-posed because the measurements $h_m$ are in $\Spc X_2' \subset \Spc H'$. 
According to Theorem \ref{theo:ART}, this problem admits a unique minimizer $g^* = \sum_{m=1}^M a_m h_m^*$, where $h_m^* \in \Spc H$. By definition, the function $h^*$ also satisfies  $\Op H \{h^*\} = (\V z_0 - \Op H \{p^*\})$. Moreover, $ \| \Op L h^* \|^2_{L_2} \leq  \| \Op L  g^*\|^2_{L_2}$; otherwise, the function $\tilde{f} = g^* + p^* \in \Spc X_2$ would satisfy $J_2(\V z | \tilde{f} ) < J_2(\V z | f^*)$, which is impossible. This means that $\tilde{f}$ is minimizing \eqref{eq:intermediateminimization}. By unicity, one has that $h^* = g^* = \sum_{m=1}^M a_m h_m^*$. 

So far, we have shown that $f^* =p^* + \sum_{m=1}^M a_m h_m^* $.
The Riesz map $ \Op R : \Spc H' \rightarrow \Spc H$ is given for $h \in \Spc H'$ by 
\begin{equation}
	\Op R \{ h \} (x) = \int_{\R} \rho_{\Op L^* \Op L}( x - y ) h(y) \mathrm{d} y = ( \rho_{\Op L^* \Op L} * h)(x),
\end{equation}
where $\rho_{\Op L^* \Op L}$ is the Green's function of the operator $(\Op L^* \Op L)$ (see Definition \ref{def:splineadmiss}). This is easily seen from the form of the norm $\lVert \Op L \cdot \rVert_{L_2}$ over $\Spc H$ and the characterization of the Riesz map as $\langle \Op R f , g \rangle_{\Spc H} = \langle f , g \rangle$. This implies that $h^*_m = \rho_{\Op L^* \Op L} * h_m = \varphi_m$ and $f^*$ has the form \eqref{eqn:spline_L2}. 

We conclude by remarking that the condition $\Op R  h  \in \Spc H$ for every $h \in \Spc H'$ implies in particular that $\sum_{m} a_m h_m \in \Spc H$, or, equivalently, that $\sum_m a_m \langle h_m , p \rangle = 0$ for every $p \in \Spc N_{\Op L}$, which proves \eqref{eqn:ortho}.

\section{Proof of Theorem \ref{theo:L1_representer}} \label{app:proofL1theo}

As for the $L_2$ case, the proof has two steps: We first show that the set of minimizers is nonempty. We then connect the optimization problem to the one studied in \cite[Theorem 2]{Unser2016} to deduce the form of the extreme points. 
The functional to minimize is $J_1(\V z |f) = E(\V z, \Op H\{f\})+ \lambda \| \Lop f\|_{\Spc M}$, defined over $f$ in the Banach space $\Spc X_1$. \\

\textit{Existence of Solutions.} We first show that $\mathcal{V} = \left\{\arg \min_{f\in \Spc X_1}  J_1(\V z |f)\right\}$ is nonempty. We use the results of Theorem \ref{theo:banachopti}, which can be found in \cite[Section 3.6]{Ito2016}.

\begin{theorem} \label{theo:banachopti}
Let $F : \Spc X \rightarrow \R^+$ be a functional on the Banach space $\Spc X$ with norm $\lVert \cdot \rVert$. 
\begin{enumerate}
    \item A convex and lower semi-continuous functional on $\Spc X$ is weakly lower semi-continuous.
    
    \item The norm $\lVert \cdot \rVert$ is weakly lower semi-continuous in $\Spc X$.
        
    \item A weakly lower semi-continuous and coercive functional on $\Spc X$ reaches its infimum. 
\end{enumerate}
\end{theorem}

According to Theorem \ref{theo:banachopti}, the existence of solutions is guaranteed if $J_1(\V z |\cdot)$ is weakly lower semi-continuous and coercive. The coercivity is deduced exactly in the same way we did for Theorem \ref{theo:L2_representer}. The continuity is obtained as follows: 
The function $E( \V z | \cdot)$ is convex and lower semi-continuous in $\R^M$ and, therefore, weakly lower semi-continuous by Theorem \ref{theo:banachopti}. 
Moreover, $\Op H$ is weak*-continuous by assumption. Hence, it is continuous for the norm topology. (Indeed, the weak*-topology being weaker than the norm topology on $\Spc X_1$, it is less restrictive to be continuous for the norm topology, that has more open sets, than for the weak*-topology.) It implies that $E(\V z | \Op H \{ \cdot \})$ is weakly lower semi-continuous by composition. 
Moreover, the norm $\lVert \cdot \rVert_{\Spc X_1}$ is lower semi-continuous on $\Spc X_1$ by Theorem \ref{theo:banachopti}. Finally, $J_1(\V z |\cdot)$ is lower semi-continuous as the sum of two lower semi-continuous functionals.

\textit{Form of the Extreme Points.}
Theorem \ref{solution set} implies that all minimizers of $J_1(\V z |\cdot)$ have the same measurement $\Op H\{f^*\} = \V z_0$. The set of minimizers is thus equal to
\begin{equation}
    \mathcal{V} = \left\{\arg \min_{f \in \Spc X_1} \| \Lop f \|_{\Spc M}, \ \text{s.t. } \Op H \{f\} = \V z_0\right\}.
\end{equation}
Since $\Spc V$ is nonempty, the condition $\Op H \{f\} = \V z_0$ is feasible. 
We can therefore apply Theorem 2 of \cite{Unser2016} to deduce that $\mathcal{V}$ is convex and weak*-compact, together with the general form \eqref{eq:L1formsolutions} of the extreme-point solutions. 

\section{Proof of Theorem \ref{prop:Simplex}}\label{proof:simplex}
We first state two propositions that are needed for the proof. Their proofs are given in the supplementary material.
\begin{proposition}[\text{Adapted from \cite[Theorem 5]{Unser2016b}}]\label{prop:LassoSol} Let $\V z \in \Rm$ and $\M H \in \R^{M \times N}$, where $M<N$. Then, the solution set $\alpha_{\lambda}$ of 
\begin{align}
\M a^*=\arg \min_{\M a \in \R^{N}} \left( \|\V z -\M H \M a\|^2_2+ \lambda \|\M a\|_1 \right)
\end{align}
is a compact convex set and $\|\M a\|_0 \leq M,\,\forall\M a \in \alpha_{E,\lambda}$, where $\alpha_{E,\lambda}$ is the set of the extreme points of $\alpha_{\lambda}$.
\end{proposition}

\begin{proposition}\label{prop:concat}
Let the convex compact set $\alpha_{\lambda}$ be the solution set of Problem (\ref{eqn:Lassolemma}) and let $\alpha_{E,\lambda}$ be the set of its extreme points. Let the operator $T:\alpha_{\lambda}\to \R^N $ be such that $\Op T \M a=\M u \text{ with } u_m= |a_m|, m \in [1\ldots N] $. Then, the operator is linear  and invertible over the domain $\alpha_{\lambda}$  and the range $\Op T\alpha_{\lambda}  $ is convex compact such that the image of any extreme point $\M a_{E}\in  \alpha_{E,\lambda} $ is also an extreme point of the set $\Op T\alpha_{\lambda}  $.
\end{proposition}

The linear program corresponding to (\ref{eqn:LP2}) is
\begin{align}\label{eqn:LP3}
 \left(\M a^*, \M u^*\right)=\min_{\M a,\M u} ~\sum_{n=1}^N  u_n, ~\text{subject~to}~
 \M u +\M a&\geq \M 0,\nonumber \\[-4mm]
 \M u -\M a&\geq \M 0,\nonumber \\ 
 \quad \M P \M a &=\V z.
 \end{align}
 By putting $ \M u +\M a=\M s_1$ and $(\M u -\M a)=\M s_2$, the standard form of this linear program is 
\begin{align}\label{eqn:LP4}
\left(\M s^*_1, \M s^*_2\right)=\min_{\M s_1,\M s_2} \left(~\sum_{n=1}^N  s_{1n}+s_{2n}\right), ~\text{s.t.}~
 \M s_1 &\geq \M 0, \nonumber\\[-4mm]
 \M s_2 &\geq \M 0, \nonumber\\
 \quad \M P \M s_1-\M P \M s_2 &\leq\V z\nonumber\\
 \quad -\M P \M s_1+\M P \M s_2 &\leq-\V z.
\end{align}
Any solution $\M a^*$ of \eqref{eqn:LP3} is equal to $(\M s^*_{1} -\M s^*_2)$ for some solution pair (\ref{eqn:LP4}). 
We denote the concatenation of any two independent points $\M s^r_1,\M s^r_2 \in \R^N$ by the variable $\M s^r=\left(\M s^r_1,\M s^r_2 \right) \in \R^{2N}$. Then, the concatenation of the feasible pairs $\M s^f=\left(\M s^f_1, \M s^f_2\right)$ that satisfies the constraints of the linear program (\ref{eqn:LP4}) forms a polytope in $\R^{2N}$. 
Given that \eqref{eqn:LP4} is solvable, it is known that at least one of the extreme points of this polytope is also a solution. The simplex algorithm is devised such that its solution $\M s^*_{\rm SLP}=\left(\M s^*_{1,\rm SLP}, \M s^*_{2,\rm SLP}\right)$ is an extreme point of this polytope \cite{Luenberger1973}. Our remaining task is to prove that $\M a_{\rm SLP}^*=\left(\M s^*_{1,\rm SLP} -\M s^*_{2,\rm SLP}\right)$ is an extreme point of the set $\alpha_{\lambda}$, the solution set of the problem \eqref{eqn:Lassolemma}.

Proposition \ref{prop:LassoSol} claims that the solution set $\alpha_{\lambda}$ of the LASSO problem  is a convex set with extreme points $\alpha_{E,\lambda}\in \R^N$. As $\alpha_{\lambda}$ is convex and compact, the concatenated set $\zeta= \{\V w \in \R^{2N}:\V w=\left(\M a^*, \M u^*\right), \M a^* \in \alpha_{\lambda}\}$ is convex and compact by Proposition \ref{prop:concat}. The transformation $\left(\M a^*,\M u^*\right)=\left(\M s^*_1 -\M s^*_2, \M s^*_1 +\M s^*_2\right)$ is linear and invertible. This means that the solution set of (\ref{eqn:LP4}) is convex and compact, too. The simplex solution corresponds to one of the extreme points of this convex compact set. \\
Since the map $\left(\M a^*,\M u^*\right)=\left(\M s^*_1 -\M s^*_2, \M s^*_1 +\M s^*_2\right)$ is linear and invertible, it also implies that an extreme point of the solution set of (\ref{eqn:LP4}) corresponds to an extreme point of  $\zeta$. Proposition \ref{prop:concat} then claims that this extreme point of $\zeta$ corresponds to an extreme point $\M a_{\rm SLP} \in \alpha_{\lambda,E}$.

\end{appendices}

\bibliographystyle{IEEEtran.bst} 
\bibliography{l1bib}

\newpage

\section*{Supplementary Material}

\subsection{Structure of the Search Spaces} \label{app:structure}

\emph{Decomposition of $\Spc X_1$ and $\Spc X_2$.} The set $\Spc X_1$ is the search space, or native space, for the gTV case.
It is defined and studied in \cite[Section 6]{Unser2016}, from which we recap the main results. Note that the same construction is at work for $\Spc X_2$, which is then a Hilbert space.

Let $\boldsymbol{p} = (p_1, \ldots , p_{N_0})$ be a basis of the finite-dimensional null space of $\Op L$. 
If $\boldsymbol{\phi} = (\phi_1 , \ldots , \phi_{N_0})$ and $\boldsymbol{p} = (p_1, \ldots , p_{N_0})$ form a biorthonormal system such that $ \langle \phi_{n_1}, p_{n_2} \rangle = \delta[n_1-n_2]$, and if $\phi_n$ is in $\Spc X’_1$, then
$\Op P f = \sum_{n=1}^{N_0} \langle f,  \phi_n \rangle p_n$ is a well-defined projector from $\Spc X_1$ to $\Spc N_{\Op L}$. The finite-dimensional null space of $\Op L$ is a Banach (and even a Hilbert) space for the norm 
\begin{equation}
    \lVert p \rVert_{\mathcal{N}_{\Op L}} = \left( \sum_{n=1}^{N_0} \langle p , \phi_n \rangle^2 \right)^{1/2}.
\end{equation}
Moreover, $f \in \Spc X_1$ is uniquely determined by $w = \Op L f \in \Spc M (\R)$ and $p = \Op P f \in \Spc N_{\Op L}$. More precisely, there exists a right-inverse operator $\Lop^{-1}_{\boldsymbol{\phi}}$ of $\Lop$ such that \cite[Theorem 4]{Unser2016} 
\begin{align} 
f = \Op L_{\boldsymbol{\phi}}^{-1} w + p.
\end{align}
In other words, $\Spc X_1$ is isomorphic to the direct sum $\Spc M(\R) \oplus \Spc N_{\Op L}$, from which we deduce that it is a Banach space for the norm \cite[Theorem 5]{Unser2016}
\begin{equation} \label{eq:normX1}
    \lVert f \rVert_{\Spc X_1} = \lVert \Op L f \rVert_{\Spc M} + \lVert \Op P f \rVert_{\Spc N_{\Op L}} = \lVert w \rVert_{\Spc M} + \lVert p \rVert_{\Spc N_{\Op L}}.
\end{equation}
\\
\emph{Predual of $\Spc X_1$.} The space $\Spc M(\R)$ is the topological dual of the space $C_0(\R)$ of continuous and vanishing functions. The space $\Spc X_1$ inherits this property: It is the topological dual of $C_{\Op L}(\R)$, defined as the image of $C_0(\R)$ by the adjoint $\Lop^*$ of $\Lop$ according to \cite[Theorem 6]{Unser2016}. 

We can therefore define a weak*-topology on $\Spc X_1$: It is the topology for which $f_n \rightarrow 0$ if $\langle f_n , \varphi \rangle \rightarrow 0$ for every $\varphi \in C_{\Op L}(\R)$. The weak*-topology is crucial to ensure the existence of solutions of \eqref{empiricalL1}; see \cite{Unser2016} for more details.

\subsection{Proof of Proposition \ref{prop:LassoSol}}
Using Lemma \ref{lemma:KKT}, it is clear that $\alpha_{\lambda}$ is also a solution set of
\begin{align}\label{eqn:minl1interpolation}
\alpha_{\lambda}= \arg \min \|\M a\|_1 \quad \text{s.t.}\quad  \M H \M a=\V z_{0,\lambda}
\end{align}
for some $\V z_{0,\lambda}$. The solution of the problem akin to (\ref{eqn:minl1interpolation}) has been discussed in \cite{Unser2016b} and is proven to be convex and compact such that the  extreme points $\alpha_{E,\lambda}$ of the convex set $\alpha_{\lambda}$ satisfy $\|\M a\|_0 \leq M$ for any $\M a \in \alpha_{E,\lambda}$.
\subsection{Proof of Proposition \ref{prop:concat}}

\begin{proof}
Let $\M \beta $ and $\gamma$ be such that $\beta_m= \min \left(0,\min_{\M a \in \alpha_{\lambda}}\text{sign}(a_m)\right)\in \{-1,0\}$ and $\gamma_m=\max\left(0, \max_{\M a \in \alpha_{\lambda}}\text{sign}(a_m)\right)\in \{0,1\}$ for $m \in [1\ldots N]$. Lemma \ref{lemma:KKT} claims that no two solutions from the solution set have different signs for their $m$th element. This means that the following statements are true:
\begin{align}\nonumber
&\{\beta_m \text{sign}(a_m)\geq0,\gamma_m \text{sign}(a_m) \geq0, ~\forall \M a \in \alpha_{\lambda} \}\\\nonumber
&\{\beta_m\neq 0 \Rightarrow \gamma_m =0,\gamma_m\neq 0 \Rightarrow \beta_m=0\}\\
&\{\beta_m +\gamma_m=0 \Rightarrow \beta_m=0, \gamma_m=0 \Rightarrow \text{sign}(a_m)=0, ~\forall \M a \in \alpha_{\lambda} \}\\
&\{\forall\M a \in \alpha_{\lambda}, a_m\neq0   \Rightarrow \beta_m +\gamma_m= \text{sign}(a_m) \}\\ 
&\{ \forall\M a \in \alpha_{\lambda}, |a_m|= (\beta_m +\gamma_m)a_m\}.\label{eqn:modulus}
\end{align} 
Statement (\ref{eqn:modulus}) shows that, for any $\M a \in \alpha_{\lambda}$, $\Op T \M a=\M R \M a$, where $\M R \in \R^{N\times N}$ is a diagonal matrix with entries $\M R_{mm}=\beta_m +\gamma_m$. Thus, the operation of $\Op T $ is linear in the domain $\alpha_{\lambda}$. Also, $\M a=\M R \M R \M a $ for $\M a \in \alpha$ implies that the operator $\Op T$ is invertible. \\
This ensures that the image of the convex compact set $\Op T \alpha_{\lambda}$ is also convex compact and the image of any extreme point $\M a_{E}\in  \alpha_{E,\lambda} $ is also an extreme point of the set $\Op T\alpha_{\lambda}$. Similarly, it can be proved that the concatenated set $\zeta= \{\V w \in \R^{2N}:\V w=\left(\M a, \Op T \M a\right), \M a \in \alpha_{\lambda}\}$ is the image of a linear and invertible concatenation operation on $\alpha$. Thus, it is convex and compact, and the image of any extreme point through the inverse operation of the concatenation $ w_{E}\in  \zeta_{E,\lambda} $ is also an extreme point of $\alpha_{\lambda}$.
\end{proof}
\end{document}